\documentclass[11pt,a4paper]{article}

\usepackage{etex}
\usepackage[utf8]{inputenc}
\usepackage[french,english]{babel}
\usepackage{amsmath}
\usepackage{amsfonts}
\usepackage{amssymb}
\usepackage{amsthm}
\usepackage{stmaryrd}

\usepackage[bitstream-charter]{mathdesign}
\usepackage[nottoc, notlof, notlot]{tocbibind}

\usepackage{stmaryrd}
\usepackage{units}
\usepackage[all]{xy}
\usepackage[a4paper]{geometry}
\usepackage{bbm}
\usepackage{color}
\usepackage{dsfont}
\usepackage{graphicx}
\usepackage{caption}
\usepackage{proof}
\usepackage{eufrak}

\usepackage{pgf}
\usepackage{tikz}
\usetikzlibrary{arrows,automata,matrix}

\newtheorem{theo}{Theorem}
\newtheorem{prop}[theo]{Proposition}
\newtheorem{defi}[theo]{Definition}
\newtheorem{lem}[theo]{Lemma}
\newtheorem{cor}[theo]{Corollary}

\author{Alexandre Goy}

\newcommand{\D}{\mathbb{D}}
\newcommand{\I}{\mathbb{I}}
\newcommand{\N}{\mathbb{N}}
\newcommand{\HKC}{\texttt{HKC}}
\newcommand{\R}{\mathbb{R}}
\newcommand{\M}{\mathcal{M}}

\newcommand{\e}{\mathfrak{e}}
\renewcommand{\P}{\mathbb{P}}
\newcommand{\gr}{\textbf}
\newcommand{\il}{\textit}
\newcommand{\itemz}{\item[$\triangleright$]}
\renewcommand{\1}{\oplus}
\newcommand{\0}{*}

\renewcommand{\ll}{\pmb{\llbracket}}
\newcommand{\rr}{\pmb{\rrbracket}}

\begin{document}

\title{Trace semantics via determinization for probabilistic transition systems}

\begin{titlepage}

\newcommand{\HRule}{\rule{\linewidth}{0.5mm}}
\begin{center}

\begin{figure}[!h]
   \begin{minipage}[c]{.46\linewidth}
      \includegraphics[scale=0.1]{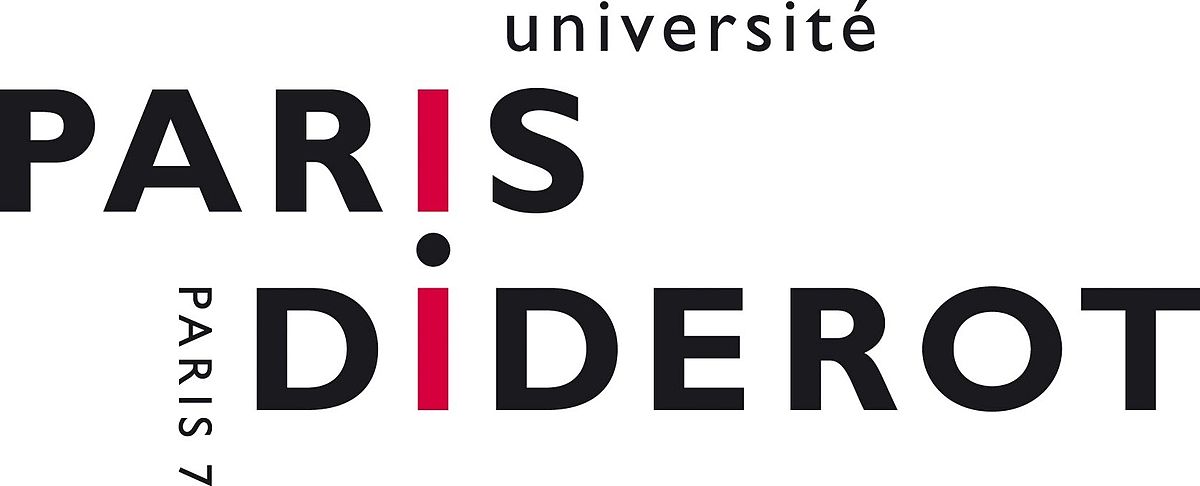}
   \end{minipage} \hfill
   \begin{minipage}[c]{.46\linewidth}
      \includegraphics[scale=0.35]{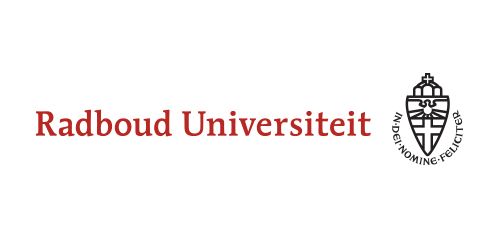}
   \end{minipage}
\end{figure}
\textsc{\Large Master 2}\\[0.5cm]
\textsc{\large Logique Mathématique et Fondements de l'Informatique}\\[0.5cm]

\HRule \\[0.4cm]
{ \huge \bfseries Trace semantics via determinization for probabilistic transition systems}\\[0.4cm]
\HRule \\[1.0cm]
 
\begin{minipage}{0.4\textwidth}
\begin{flushleft} \large
\emph{Author:}\\
Alexandre \textsc{Goy}
\end{flushleft}
\end{minipage}
~
\begin{minipage}{0.4\textwidth}
\begin{flushright} \large
\emph{Supervisor:} \\
Dr. Jurriaan \textsc{Rot}
\end{flushright}
\end{minipage}\\[0.5cm]

{\large \today}\\[0.5cm]
\tableofcontents
\vfill
\end{center}
\end{titlepage}
\newpage

\abstract
A coalgebraic definition of finite and infinite trace semantics for probabilistic transition systems has recently been given using a certain Kleisli category. In this paper this semantics is developed using a coalgebraic method which is an instance of general determinization. Once applied to discrete systems, this point of view allows the exploitation of the determinized structure by up-to techniques. Thereby it becomes possible to algorithmically check the equivalence of two finite probabilistic transition systems.

\section*{Introduction}
\addcontentsline{toc}{section}{Introduction}  

Automata theory is certainly one of the most explored branches of computer science. To meet the growing needs in probabilistic programming, model checking, or randomized algorithms, many kinds of automata are shaped with a probabilistic behaviour. Here is one of them: a generative probabilistic transition system (PTS) consists of a state space $X$, where every state $x$ can either terminate or transition to any state, with a certain probability. Each non-terminating transition outputs a letter $a$ from an alphabet $A$. With this informal definition, one can see that given a state $x$, each word $w$ is generated by the automaton with a certain probability $\llbracket x \rrbracket (w)$. The function $w \mapsto \llbracket x \rrbracket (w)$ is itself a probability distribution if we take into account both finite and infinite words. The aim of this paper is to study the formal definition of these semantics, referred to as the \il{trace semantics}.\\\\
Automata will be described as usual using graphs. Each state is pictured by a circle and there is a distinguished terminal state $*$ which is double-circled. A transition is represented by an arrow labeled with its probability. Non-terminating arrows are further labeled with one transition letter. Consider as a first example the following PTS:
\begin{center}
\begin{tikzpicture}[->,>=stealth',shorten >=1pt,auto,node distance=2.8cm,
                    semithick]
  \tikzstyle{every state}=[fill=none,draw=black,text=black]

  \node[state](A)                    {$x$};
  \node[state,accepting](B) [right of=A] {$*$};

  \path (A) edge  [loop left] node {a,1/2} (A)
            edge              node {1/2} (B);
\end{tikzpicture}
\end{center}
One can intuitively associate trace semantics to this automaton. Given the above PTS, the only reasonable semantics is $\llbracket x \rrbracket(a^n) = \frac{1}{2^{n+1}}$. It is clearly a probability measure over $A^*$. Actually, some strange phenomena can occur if $|A|\geq 2$ because of the fact that $A^\mathbb{N}$ is uncountable. Consider the following example where $A = \{a,b\}$.
\begin{center}
\begin{tikzpicture}[->,>=stealth',shorten >=1pt,auto,node distance=2.8cm,
                    semithick]
  \tikzstyle{every state}=[fill=none,draw=black,text=black]

  \node[state](A)                    {$y$};
  \path (A) edge  [loop left] node {a,1/2} (A)
  		(A) edge  [loop right] node {b,1/2} (A);
\end{tikzpicture}
\end{center}
No matter how you look at it, you should end up with $\llbracket y \rrbracket(w) = 0$ for every finite or infinite word $w$. A possible way to fix this is to ask for a subprobability over words instead of a probability. Then $\llbracket y \rrbracket$ can be defined as a subprobability with total mass $0$. But a serious problem appears, because with this definition, the following state $z$ has the same (trivial) semantics as $y$, so that they are deemed equivalent.
\begin{center}
\begin{tikzpicture}[->,>=stealth',shorten >=1pt,auto,node distance=2.8cm,
                    semithick]
  \tikzstyle{every state}=[fill=none,draw=black,text=black]

  \node[state](A)                    {$z$};
  \path (A) edge  [loop left] node {a,3/4} (A)
  		(A) edge  [loop right] node {b,1/4} (A);
\end{tikzpicture}
\end{center}
But for example, $y$ is twice as likely as $z$ to generate an infinite word that begins with $b$. This is a concrete behavioral difference. In case infinite traces are taken into account, both states are thus not to be considered equivalent. Further, we would like to have techniques to prove that they are (not). To this purpose, it is necessary to dive into measure theory and define the arguments of $\llbracket x \rrbracket$ to be \il{sets of words} instead of being simply words.\\\\
To this end we will use \il{coalgebras}. State-based sytems are increasingly modeled using a coalgebraic point of view, which benefits from the powerful toolbox of category theory. An introduction to coalgebras can be found in \cite{Jacobs97} or \cite{Jacobs16}. This framework is notably convenient when it comes to define trace semantics; this is performed for very general PTS in \cite{Kerstan13}, using a construction in a Kleisli category as in \cite{Hasuo07}. The usual counterpart of this kind of construction is to carry out a determinization process formally based on an Eilenberg-Moore category. Concretely, the state space is changed in order to make transitions become deterministic.  The comparison between these two methods is discussed in \cite{Jacobs15}. Determinization allows to exploit \il{bisimulation up-to} techniques, a family of proof methods for behavioural equivalence of state-based systems, which have been extensively applied in concurrency theory and, more recently, in automata theory. For example, non-deterministic automata can be determinized via the well-known powerset construction. This opens the way for the $\texttt{HKC}$ algorithm of \cite{Bonchi13} that checks the equivalence of non-deterministic automata using up-to techniques. \\\\
Our main contribution is to redefine the Kleisli trace semantics of \cite{Kerstan13} using the Eilenberg-Moore method. This is done for both discrete and continuous sytems. In the case of discrete systems, our approach allows to generalize the $\HKC$ algorithm of \cite{Bonchi17} to an algorithm $\HKC^\infty$ that checks equivalence of states (i.e. it checks if $\llbracket x \rrbracket = \llbracket y \rrbracket $) for both finite and infinite words. Our paper is organized as follows. In section $1$ are introduced the basic concepts of category theory, coalgebraic modeling and measure theory, including a part about measurable sets of words. Section $2$ deals with the discrete case. The determinization process is performed  by hand in order to go straight to the $\HKC^\infty$ algorithm and its correctness. This is followed by a few examples. The general setting is presented in section $3$, where the origin of the determinization process is further explained. Its central result is Theorem \ref{keybis}. Section $3$ ends with the proof that our semantics is the same as the Kleisli semantics from \cite{Kerstan13}; the general framework that relates both constructions is mentioned. It might well be worth to begin by scanning section $2$ to see the underlying ideas of the construction without caring about the measurability of every function, as well as some concrete brightening examples. Definitions and theorems that are not explicitely stated or referenced are considered as folklore and can be found in any basic related book.\vfill
\noindent \gr{Related work}. Our main source for the overall spirit of the whole paper is undoubtedly \cite{Kerstan13}, which is in turn influenced by the Kleisli constructions done in \cite{Hasuo07}. Both versions of $\texttt{HKC}$ presented in \cite{Bonchi13} and \cite{Bonchi17} are good starting points to understand bisimulation up-to and the determinization method. The overhanging link between Kleisli and determinization is discussed in \cite{Jacobs15}.
\newpage
\section{Preliminaries}
The set of positive integers is $\N$. The set of real numbers is $\mathbb{R}$, the set of non-negative real numbers is $\mathbb{R}_+$. The notation $1$ stands for a set with only one element $1 = \{ * \}$. Given a set $X$, its power set is $\mathcal{P}(X)$.
\subsection{Measure theory}
Let $X$ be any set. A \il{$\sigma$-algebra} on $X$ is a subset $\Sigma_X \subseteq \mathcal{P}(X)$ such that $\emptyset \in \Sigma_X$ and $\Sigma_X$ is closed under complementation and countable union. Note that this implies that $X \in \Sigma_X$ and that $\Sigma_X$ is closed under countable intersection and under set difference. Given any subset $G \subseteq \mathcal{P}(X)$, there always exists a smallest $\sigma$-algebra containing $G$. We call it the $\sigma$-algebra \il{generated} by $G$ and denote it by $\sigma_X(G)$. For example, $\mathcal{P}(X)$ is a $\sigma$-algebra on $X$. When working with real numbers $\mathbb{R}$, we will use the Borel $\sigma$-algebra $\mathcal{B}(\mathbb{R}) = \sigma_\mathbb{R} (\{ (-\infty,x] \mid x \in \mathbb{R} \})$. More specifically we denote the line segment $[0,1]$ by $\I$ and use $\mathcal{B}(\I) = \{ B \cap \I \mid B \in \mathcal{B}(\mathbb{R}) \}$ as the canonical $\sigma$-algebra on $\I$. If $X$ is a set and $\Sigma_X$ is a $\sigma$-algebra on $X$, the pair $(X,\Sigma_X)$ is called a measurable space. From now we will write $X$ for $(X,\Sigma_X)$ when the $\sigma$-algebra used is clear.\\\\
\il{Product.} Given measurable spaces $(X,\Sigma_X)$ and $(Y,\Sigma_Y)$, we define a product $\sigma$-algebra on $X \times Y$ by $\Sigma_X \otimes \Sigma_Y = \sigma_{X \times Y} (\{ S_X \times S_Y \mid S_X \in \Sigma_X, S_Y \in \Sigma_Y \})$. The product of measurable spaces is then defined by $(X,\Sigma_X) \otimes (Y,\Sigma_Y) = (X \times Y, \Sigma_X \otimes \Sigma_Y)$. Note that if $X \in G_X \subseteq \mathcal{P}(X)$ and $Y \in G_Y \subseteq \mathcal{P}(Y)$, then $\sigma_{X\times Y}(\{ S_X \times S_Y \mid S_X \in G_X, S_Y \in G_Y \}) = \sigma_X(G_X) \otimes \sigma_Y(G_Y)$.\\\\
\il{Sum.} Given measurable spaces $(X,\Sigma_X)$ and $(Y,\Sigma_Y)$, we define a sum $\sigma$-algebra on the disjoint union $X+Y = \{ (x,0) \mid x \in X \} \cup \{ (y,1) \mid y \in Y \}$ by $\Sigma_X \oplus \Sigma_Y = \{ S_X + S_Y \mid S_X \in \Sigma_X, S_Y \in \Sigma_Y \}$. The sum of measurable spaces is then defined by $(X,\Sigma_X)\oplus (Y,\Sigma_Y) = (X+Y,\Sigma_X \oplus \Sigma_Y)$. Note that if $\emptyset \in G_X \subseteq \mathcal{P}(X)$ and $Y \in G_Y \subseteq \mathcal{P}(Y)$, then $\sigma_{X+Y}(G_X \oplus G_Y) = \sigma_X (G_X) \oplus \sigma_Y (G_Y)$, see \cite{Kerstan13}.\\\\
Binary products and sums can be easily generalized to finite products and sums by induction. These are denoted by $\bigotimes_{i\in I}$ and $\bigoplus_{i\in I}$ respectively.\\\\
A function $f : (X,\Sigma_X) \to (Y,\Sigma_Y)$ is \textit{measurable} if for all $S_Y \in \Sigma_Y$, $f^{-1}(S_Y) \in \Sigma_X$. The composition of measurable functions is measurable.
\begin{lem} \label{generated}
If $f : (X,\Sigma_X) \to (Y,\sigma_Y(G_Y))$ is such that for all $S_Y \in G_Y$, $f^{-1}(S_Y) \in \Sigma_X$, then it is measurable.
\end{lem}
\begin{proof}
The proof is given as it shows the classical way of reasoning about generated $\sigma$-algebras. Let $\Sigma = \{ S \subseteq Y \mid f^{-1}(S) \in \Sigma_X \} \subseteq \mathcal{P}(Y)$. Note that $\emptyset \in \Sigma$ because $f^{-1}(\emptyset) = \emptyset \in \Sigma_X$, and that $\Sigma$ is closed under complementation and countable union. Indeed, if $(S_n)_{n\in \N} \in \Sigma^\N$ are such that $f^{-1}(S_n) \in \Sigma_X$ for all $n\in \N$, we have $f^{-1}(Y\setminus S_0) = X\setminus f^{-1}(S_0) \in \Sigma_X$ as $\Sigma_X$ is closed under complementation, and $f^{-1}\left( \bigcup_{n\in \N} S_n \right) = \bigcup_{n\in\N} f^{-1}(S_n) \in \Sigma_X$ as $\Sigma_X$ is closed under countable unions. Thus $\Sigma$ is a $\sigma$-algebra. Furthermore, the hypothesis gives that $G_Y \subseteq \Sigma$. As $\sigma_Y(G_Y)$ is the \il{smallest} $\sigma$-algebra that contains $G_Y$, we get $\sigma_Y(G_Y) \subseteq \Sigma$ so that $f$ is measurable.
\end{proof}
\noindent For any family of measurable spaces $(Z_i,\Sigma_{Z_i})_{i\in I}$, and any family of functions $(f_i)_{i \in I} : Y \to Z_i$, define $\Sigma_Y$ as the smallest $\sigma$-algebra on $Y$ that makes every $f_i$ measurable. This is the $\sigma$-algebra \il{generated} by $(f_i)_{i\in I}$.
\begin{lem} \label{usb}
Given $h : (X,\Sigma_X) \to (Y,\Sigma_Y)$ where $\Sigma_Y$ is generated by some $(f_i)_{i\in I}$, assume that for all $i \in I$, $(f_i \circ h)$ is measurable. Then $h$ is measurable.
\end{lem}
\begin{proof}
First prove that $\Sigma_Y = \sigma_Y(\{ f_i^{-1}(S_i)\mid i\in I, S_i \in \Sigma_{Z_i} \})$. The $\sigma$-algebra on the right makes all functions $f_i$ measurable, thus it contains $\Sigma_Y$. On the other hand, for all $i\in I$, $S_i \in \Sigma_{Z_i}$, we have $f_i^{-1}(S_i) \in \Sigma_Y$ because $f_i$ is measurable. Now apply Lemma \ref{generated} using that for all $i \in I$ and $S_i \in \Sigma_{Z_i}$, $h^{-1}(f_i^{-1}(S_i)) = (f_i \circ h)^{-1}(S_i) \in \Sigma_X$ because $f_i \circ h$ is measurable.
\end{proof}
\noindent For instance, using Lemma \ref{usb} facilitates measurability proofs for functions with several components. Let $(Z_i,\Sigma_i)_{i \in I}$ be a finite sequence of measurable spaces, $(X,\Sigma_X)$ be a measurable space, and $f : (X,\Sigma_X) \to \left(\prod_{i\in I} Z_i, \bigotimes_{i\in I} \Sigma_i\right)$ be a function with components $(f_i : X \to Z_i)_{i\in I}$. Note that $\bigotimes_{i\in I} \Sigma_i$ is the smallest $\sigma$-algebra on $\prod_{i\in I} Z_i$ that makes every projection $\pi_j : \prod_{i\in I} Z_i \to Z_j $ measurable. (Indeed, $\pi_j^{-1}(S_j) = \prod_{i\in I} X_i^\varepsilon$ where $X_j^\varepsilon = S_j$ and $X_i^\varepsilon = X_i$ for $i \neq j$, so that $\bigotimes_{i\in I} \Sigma_i$ makes every $\pi_j$ measurable; on the other hand, a generator of $\prod_{i\in I} Z_i$ can be written $\prod_{i\in I} S_i = \bigcap_{i\in I} \pi^{-1}(S_i)$ which is in any $\sigma$-algebra that makes every $\pi_j$ measurable.) Thus applying Lemma \ref{usb} it is sufficient that every $f_j$ is measurable in order for $f$ to be measurable.\\\\
A \il{finite measure} on the measurable space $(X,\Sigma_X)$ is a function $m : \Sigma_X \to \mathbb{R}_+$ such that $m(\emptyset)=0$ and if $(S_n)_{n\in\N} \in \Sigma_X^\N$ are disjoint, then $m\left(\bigcup_{n\in\N} S_n \right) = \sum_{n\in \N} m(S_n)$ (this is called \il{$\sigma$-additivity}). In this case the triple $(X,\Sigma_X,m)$ is called a \il{measure space}. If $m(X)=1$ (resp. $m(X)\leq 1)$ it is a \il{probability space} and $m$ is a \il{probability measure} (resp. sub-probability). Finite measures have pleasant properties: $m(X\setminus A) = m(X) - m(A)$ for any $A \in \Sigma_X$, $m\left( \bigcup_{n\in \N} A_n \right) = \lim_{n\to \infty} m(A_n)$ for any increasing sequence $(A_n)_{n\in\N}$ of $\Sigma_X$, and analogously $m\left(\bigcap_{n\in \N} B_n \right) = \lim_{n\to \infty} m(B_n)$ for any decreasing sequence $(B_n)_{n\in\N}$ of $\Sigma_X$. These measures are called finite because they satisfy $m(X) < \infty$. 
\subsubsection*{Measure over words}
We begin with some reminders about languages. Any finite set $A$ can be called an \il{alphabet} and its elements \il{letters}. The set of words of length $n$ with letters in $A$ is denoted by $A^n$. By convention $A^0 =\{ \varepsilon \}$ where $\varepsilon$ is the empty word. The set of finite words over $A$ is denoted by $A^* = \bigcup_{n\in\N} A^n$, the set of infinite words by $A^\omega = A^\N$ and the set of (finite and infinite) words by $A^\infty = A^* \cup A^\omega$. A language $L$ is a subset of $\mathcal{P}(A^*)$. It can be seen as a function $L : A^* \to \{0,1\}$, by setting $L(w)=1$ iff $w\in L$. The \il{language derivative} of $L$ with respect to a letter $a$ is defined by $L_a(w)=L(aw)$.\\\\
The length of $w \in A^\infty$ is denoted by $|w| \in \N \cup \{ \infty \}$. The concatenation function $A^* \times A^\infty \to A^\infty$ is denoted by juxtaposition and defined by $uv(n) = u(n)$ if $n < |u|$ and $uv(n) = v(n-|u|)$ if $|u|\leq n < |u| + |v|$. It can be extended to languages $\mathcal{P}(A^*) \times\mathcal{P}(A^\infty) \to \mathcal{P}(A^\infty)$ by setting $LM = \{ uv \mid u \in L, v \in M \}$. In the following, $\{w\} M$ will be denoted by $wM$ for convenience. \\\\
One aim of this paper is to define certain probability measures on $A^\infty$ which reflects the behaviour of a given automaton. For this purpose we need to make precise which $\sigma$-algebra we use.
\begin{defi}
Let $S_\infty = \{\emptyset\} \cup \{ \{w\} \mid w \in A^* \} \cup \{wA^\infty \mid w \in A^* \}$. Define the $\sigma$-algebra of \emph{measurable sets of words} to be $\Sigma_{A^\infty} = \sigma_{A^\infty}(S_\infty)$.
\end{defi}
\noindent This $\sigma$-algebra is generated by a countable family of simple generators: the empty set, the singletons of finite words, and the cones, i.e., sets $wA^\infty$ of words that have the finite word $w$ as a prefix. In the sequel, this is the $\sigma$-algebra on $A^\infty$ implicitly used. As a first step, we take a look at some properties of measurable sets of words.
\begin{prop}\label{usualsets} The following sets of words are measurable:
\begin{enumerate}
\item[(i)] The singleton $\{w\}$ for any $w \in A^\infty$;
\item[(ii)] Any countable language;
\item[(iii)] Any language of finite words;
\item[(iv)] $\emptyset$, $A^*$, $A^\omega$, $A^\infty$;
\item[(v)] The concatenation $LS$ where $L \subseteq A^*$ and $S \in \Sigma_{A^\infty}$.
\end{enumerate}
\end{prop}
\begin{proof} \begin{enumerate}
\item[$(i)$] It is already known for finite words $w$ since $\{w\} \in S_\infty \subseteq \Sigma_{A^\infty}$. Let $w \in A^\omega$. Let $w_{|n}$ be the finite word of length $n$ defined by $w_{|n}(k) = w(k)$ for any $k \leq n$. Then $\{w\} = \bigcap_{n\in\N} w_{|n}A^\infty \in \Sigma_{A^\infty}$ because $\Sigma_{A^\infty}$ is closed under countable intersections.
\item[$(ii)$] Any countable language is a countable union of singletons, which are all measurable sets.
\item[$(iii)$] A language of finite words is a subset of the countable set $A^*$, so it is countable, hence measurable.
\item[$(iv)$] The empty set and $A^\infty$ are in $\Sigma_{A^\infty}$ because this is a $\sigma$-algebra on $A^\infty$. The set $A^*$ is a language of finite words. Finally, $A^\omega = A^\infty \setminus A^* \in \Sigma_{A^\infty}$ because $\Sigma_{A^\infty}$ is closed under complementation.
\item[$(v)$] As $LS = \bigcup_{w\in L} wS$ is a countable union, it is sufficient to prove that for all $S \in \Sigma_{A^\infty}$, $wS$ is measurable for any $w \in A^*$. Let $\Sigma = \{ S  \subseteq A^\infty \mid \forall w \in A^*, wS \in \Sigma_{A^\infty}\}$. It is easy to see that $S_\infty \subseteq \Sigma$, because for any $w \in A^*$, $w\emptyset = \emptyset \in \Sigma_{A^\infty}$, $w\{u\} = \{wu\} \in \Sigma_{A^\infty}$ and $w(uA^\infty) = (wu)A^\infty \in \Sigma_{A^\infty}$. In particular $\emptyset \in \Sigma$. Moreover, if $(S_n)_{n\in\N}$ are all in $\Sigma$, then for any $w \in A^*$ we have $w\bigcup_{n\in\N} S_n = \bigcup_{n\in\N} wS_n \in \Sigma_{A^\infty}$ since $\Sigma_{A^\infty}$ is closed under countable union, and $w(A^\infty \setminus S_0) = wA^\infty \setminus wS_0 \in \Sigma_{A^\infty}$ since $\Sigma_{A^\infty}$ is closed under set difference. Thus $\bigcup_{n\in\N} S_n \in \Sigma$ and $A^\infty \setminus S_0 \in \Sigma$ so $\Sigma$ is a $\sigma$-algebra. Since it contains $S_\infty$, we get that $\Sigma_{A^\infty} \subseteq \Sigma$ and this is exactly what had to be proved.
\end{enumerate}
\end{proof}
\noindent In the following, if $m$ is a measure over $A^\infty$ and $w \in A^\infty$, we will write $m(w)$ instead of $m(\{w\})$. Since $S_\infty$ has a pleasant structure (it is a \il{covering semiring of sets}), we have the following key theorem:
\begin{theo}[from \cite{Kerstan13}] \label{key}
Let $m : S_\infty \to \mathbb{R}_+$ be a map satisfying $m(\emptyset)=0$. The two following conditions are equivalent.
\begin{enumerate}
\item[(i)] There exists a unique measure $\tilde{m} : \Sigma_{A^\infty} \to \mathbb{R}_+$ such that $\tilde{m}_{|S_\infty} = m$.
\item[(ii)] For all $w \in A^*$, $m(wA^\infty) = m(w) + \sum_{a\in A} m(waA^\infty)$.
\end{enumerate}
\end{theo}
\begin{proof}
$(i)\Rightarrow (ii)$ This is obvious because the equation comes directly from the $\sigma$-additivity of $\tilde{m}$. $(ii)\Rightarrow (i)$ According to Lemma $3.18$ of \cite{Kerstan13}, $(ii)$ is equivalent to the fact that $m$ is a pre-measure. Using the extension theorem (Proposition $2.4$ in \cite{Kerstan13}), this pre-measure can be uniquely extended to a measure as in $(i)$.
\end{proof}
\noindent The notion of language derivative was defined above by $L_a(w) = L(aw)$. One can introduce the same operation for measures over words.
\begin{defi}[Measure derivative]
Let $m$ be a measure on $A^\infty$ and $a\in A$. The map $m_a$ defined by $m_a(S) = m(aS)$ for any $S \in \Sigma_{A^\infty}$ is a measure, called the \emph{measure derivative} of $m$ (with respect to $a$).
\end{defi}
\noindent This is well-defined because sets of the shape $aS$ for $S\in \Sigma_{A^\infty}$ are measurable according to Proposition \ref{usualsets} $(v)$, and $m_a(\emptyset) = m(a\emptyset) = m(\emptyset) = 0$ and $m_a\left( \bigcup_{n\in\N} A_n\right) = m\left(\bigcup_{n\in\N} aA_n\right) = \sum_{n\in\N} m(aA_n) = \sum_{n\in\N} m_a(A_n)$ if the $A_n$ are disjoint.\\\\
In the following, the space of sub-probability measures on $(A^\infty,\Sigma_{A^\infty})$ is denoted by $\M(A^\infty)$. Note that $\M (A^\infty)$ is closed under measure derivatives: if $m(A^\infty) \leq 1$ then $m_a(A^\infty) = m(aA^\infty) \leq m(A^\infty) \leq 1$.
\subsection{Category theory}
We recall here some basics of category theory. A \il{category} consists of a class of objects $\gr{C}$ and a class of morphisms (or arrows) $\gr{C}(X,Y)$ for every objects $X,Y$ of $\gr{C}$. The notation $f : X \to Y$ stands for the sentence "$f$ is an arrow of $\gr{C}(X,Y)$". For each object $X$ of $\gr{C}$ there is an identity morphism $id_X : X \to X$. Furthermore, there is a composition function $\circ$  which is associative i.e. $(h\circ g)\circ f = h \circ (g\circ f)$, and such that if $f : X \to Y$, then $f \circ id_X = f = id_Y \circ f$. The composition $g \circ f$ is possible iff there exists objects $X,Y,Z$ such that $f : X \to Y$ and $g : Y \to Z$. In the following, we will mainly work with the two following categories:
\begin{itemize}
\itemz The category $\gr{Sets}$ of sets and functions. Objects are usual sets, morphisms are functions, identity morphisms are identity functions and composition is given by the usual composition of functions.
\itemz The category $\gr{Meas}$ of measurable sets and functions. Objects are \il{measurable spaces} $(X,\Sigma_X)$. Morphisms are \il{measurable functions} $f : (X,\Sigma_X) \to (Y,\Sigma_Y)$. Identity morphisms are identity functions and composition is given by the usual composition of functions.
\end{itemize}
Let $\gr{C}$ be a category and $A$ be an object of $\gr{C}$. The object $A$ is \il{final} if for all object $X$ of $\gr{C}$ there exists a unique morphism $f_X : X \to A$ called the \il{final morphism}. Such an object is unique up to (unique) isomorphism.\\\\
A \il{functor} from $\gr{C}$ to $\gr{D}$ is a mapping $F$ that associates to every object $X$ of $\gr{C}$ an object $FX$ of $\gr{D}$, to every morphism $f : X \to Y$ a morphism $Ff : FX \to FY$, and such that $F(id_X) = id_{FX}$ and $Fg \circ Ff = F(g\circ f)$. An endofunctor of $\gr{C}$ is a functor $F : \gr{C} \to \gr{C}$. The composition of two functors is still a functor. For example, the identity functor $Id_\gr{C} : \gr{C} \to \gr{C}$ maps each object and morphism to itself. Given an object $X$ of $\gr{D}$, the constant functor $X : \gr{C} \to \gr{D}$ maps each object to $X$ and each morphism to $id_X$.\\\\
Let $F,G : \gr{C} \to \gr{D}$ be some functors. A \il{natural transformation} $\lambda : F \Rightarrow G$ consists of a $\gr{D}$-morphism $\lambda_X : FX \to GX$ for every $X \in \gr{C}$, such that for every $\gr{C}$-morphism $f : X \to Y$, the following diagram commutes.
\begin{center}
\begin{tikzpicture}
  \matrix (m) [matrix of math nodes,row sep=3em,column sep=4em,minimum width=2em]
  {
     FX & GX \\
     FY & GY \\};
  \path[-stealth]
    (m-1-1) edge node [above] {$\lambda_X$} (m-1-2)
    (m-1-2) edge node [right] {$Gf$}(m-2-2)
    (m-1-1) edge node [left] {$Ff$} (m-2-1)
    (m-2-1) edge node [below] {$\lambda_Y$}(m-2-2);
\end{tikzpicture}
\end{center}
Given a category $\gr{C}$, a \textit{monad} is a triple $(T,\eta,\mu)$ where $T : \gr{C} \to \gr{C}$ is an endofunctor and $\eta : Id_\gr{C} \Rightarrow T$, $\mu : TT \Rightarrow T$ are natural transformations called \il{unit} and \il{multiplication} respectively, such that the following two diagrams commute.
\begin{center}
\begin{tikzpicture}
  \matrix (m) [matrix of math nodes,row sep=3em,column sep=4em,minimum width=2em]
  {
     TX & TTX & & TTTX & TTX \\
     TTX & TX & & TTX & TX \\};
  \draw[double equal sign distance] (m-1-1) -- (m-2-2);
  \path[-stealth]
    (m-1-1) edge node [above] {$\eta_{TX}$} (m-1-2)
    (m-1-2) edge node [right] {$\mu_X$}(m-2-2)
    (m-1-1) edge node [left] {$T\eta_X$} (m-2-1)
    (m-2-1) edge node [below] {$\mu_X$}(m-2-2)
    (m-1-4) edge node [above] {$\mu_{TX}$} (m-1-5)
    (m-2-4) edge node [below] {$\mu_X$} (m-2-5)
    (m-1-4) edge node [left] {$T\mu_X$} (m-2-4)
    (m-1-5) edge node [right] {$\mu_X$} (m-2-5);
\end{tikzpicture}
\end{center} 
\textbf{Examples.} In the context of coalgebras, monads are often used to model branching behaviour, like non deterministic branching or probabilistic branching. It will always be clear from the name of the functor which monad is used. Hence, units are always denoted by $\eta$ and multiplications by $\mu$. \begin{itemize}
\itemz In $\gr{Sets}$, the power set monad $(\mathcal{P},\eta,\mu)$ is defined as follows. Given two sets $X,Y$ and a function $f : X \to Y$, $\mathcal{P}X$ is the power set $\mathcal{P}(X)$ and $\mathcal{P} f : S \in \mathcal{P}X \mapsto f(S) \in \mathcal{P}Y$ is the direct image. The unit is the singleton $\eta_X(x) = \{ x \}$ and the multiplication is given by $\mu_X(\mathcal{S}) = \bigcup_{S\in \mathcal{S}} S$.
\itemz In $\gr{Sets}$, the probability monad $(P,\eta,\mu)$ is defined as follows. For any $u : X \to \I$ the \il{support} of $u$ is defined by $\text{\il{supp}}(u) = \{ x \in X \mid u(x) \neq 0 \}$. If $\text{\il{supp}}(u)$ is finite we say that $u$ has finite support.  Given two sets $X,Y$ and a function $f : X \to Y$, define $PX = \left\{ u : X \to \I \mid u \text{ has finite support and } \sum_{x\in X} u(x) = 1 \right\}$ and $Pf(u)(y) = \sum_{x\in f^{-1}(\{y\})} u(x)$. The unit is the Kronecker delta $\eta_X(x)(y) = \delta_{x,y}$ and the multiplication is given by $\mu_X(U)(y) = \sum_{u\in PX} U(u)u(y)$. This monad models finitely branching probabilistic behaviour. A variant is the subprobability monad $(D,\eta,\mu)$. The only difference is that the sum may be less than $1$: $DX = \left\{ u : X \to \I \mid u \text{ has finite support and } \sum_{x\in X} u(x) \leq 1 \right\}$.
\end{itemize}

\subsubsection*{Algebras and distributive laws}
Let $F : \gr{C} \to \gr{C}$ be an endofunctor. An \textit{$F$-algebra} is an object $X$ together with an arrow $\alpha : FX \to X$. An $F$-algebra can be viewed as an operation that takes some information out of $FX$ and constructs a new element in $X$. For example, in $\gr{Sets}$, if some set $S$ is given and $FX = S \times X$, $Ff = Id \times f$, then the arrow $\alpha : S \times S^\omega \to S^\omega$ given by concatenation $\alpha(h,t) = h \cdot t$ can be thought of as a way of making a list from a head element $h$ and a tail list $t$.
A morphism of $F$-algebras $\alpha : FX \to X$ and $\beta : FY \to Y$ is an arrow $f : X \to Y$ such that $f \circ \alpha = \beta \circ Ff$. The collection of $F$-algebras and their morphisms form a category denoted by $\gr{Alg}(F)$.\\\\
The last section of this paper is committed to make clear why the construction of trace semantics is canonical. Showing that it comes from a \il{distributive law} is a convenient way to find this out. A distributive law of a monad over a functor formally allows to reverse the order in which they are applied. It is based on a variant of algebras named Eilenberg-Moore algebras. \\\\ Let $(T,\eta,\mu)$ be a monad on a category $\gr{C}$. An \il{Eilenberg-Moore $T$-algebra} is a $T$-algebra $\alpha$ such that the following diagrams commute.
\begin{center}
\begin{tikzpicture}
  \matrix (m) [matrix of math nodes,row sep=3em,column sep=4em,minimum width=2em]
  {
     X & TX & & TTX & TX \\
     & X & & TX & X \\};
  \path[-stealth]
  (m-1-1) edge node [above] {$\eta_X$} (m-1-2)
  (m-1-1) edge node [below left] {$id_X$} (m-2-2)
  (m-1-2) edge node [right] {$\alpha$} (m-2-2)
  (m-1-4) edge node [above] {$\mu_X$} (m-1-5)
  (m-1-4) edge node [left] {$T\alpha$} (m-2-4)
  (m-1-5) edge node [right] {$\alpha$} (m-2-5)
  (m-2-4) edge node [below] {$\alpha$} (m-2-5)
	;
\end{tikzpicture}
\end{center}
Given two Eilenberg-Moore $T$-algebras $\alpha : TX \to X$ and $\beta : TY \to Y$, a morphism from $\alpha$ to $\beta$ is a $\gr{C}$-arrow $f : X \to Y$ such that $f \circ \alpha = \beta \circ Tf$. Eilenberg-Moore $T$-algebras and their morphisms form a category denoted by $\gr{EM}(T)$.\\\\
Let $F : \gr{C} \to \gr{C}$ be an endofunctor. A natural transformation $\lambda : TF \Rightarrow FT$ is a \il{distributive law} if the following diagrams commute for every object $X$ of $\gr{C}$.
\begin{center}
\begin{tikzpicture}
  \matrix (m) [matrix of math nodes,row sep=3em,column sep=4em,minimum width=2em]
  {
     FX & FX & & TTFX & TFTX & FTTX \\
     TFX & FTX & & TFX & & FTX \\};
  \path[-stealth]
  (m-1-1) edge node [above] {$id_X$} (m-1-2)
  (m-1-1) edge node [left] {$\eta_{FX}$} (m-2-1)
  (m-1-2) edge node [right] {$F\eta_X$} (m-2-2)
  (m-2-1) edge node [below] {$\lambda_X$} (m-2-2)
  (m-1-4) edge node [above] {$T\lambda_X$} (m-1-5)
  (m-1-5) edge node [above] {$\lambda_{TX}$} (m-1-6)
  (m-1-4) edge node [left] {$\mu_{FX}$} (m-2-4)
  (m-1-6) edge node [right] {$F\mu_X$} (m-2-6)
  (m-2-4) edge node [below] {$\lambda_X$} (m-2-6)
	;
\end{tikzpicture}
\end{center}
\subsubsection*{Coalgebras}
The dual concept of $F$-algebra is that of an \il{$F$-coalgebra}, i.e., an object $X$ together with an arrow $\alpha : X \to FX$. An $F$-coalgebra can be seen as a way of observing some behavior. For example, if $FX = S \times X$ as above, $\alpha : S^\omega \to S \times S^\omega$ given by $\alpha(h\cdot t) = (h,t)$ is a system which decomposes an infinite list into its head and its tail. The intuition behind coalgebras is that they model systems. The functor $F$ captures the structure of the system. The object $X$ captures the state space and the arrow $\alpha$ captures both the transition behaviour and the observations. It will be clearer in  the end of this section, where classical deterministic or non-deterministic automata will be regarded as coalgebras. For more details about the basics intuition of coalgebras, see \cite{Jacobs97}. A morphism of $F$-coalgebras $\alpha : X \to FX$ and $\beta : Y \to FY$ is an arrow $f : X \to Y$ such that $Ff \circ \alpha = \beta \circ f$. For a given functor $F$, the collection of $F$-coalgebras and their morphisms form a category denoted by $\gr{Coalg}(F)$. This category may have a final object, called the final coalgebra, which is very interesting because it yields a canonical notion of semantics. For example, taking $F$ as above, the final $F$-algebra is $\alpha : S^\omega \to S \times S^\omega$. This means that for any $F$-coalgebra $\beta : X \to S \times X$, there exists a unique coalgebra morphism $\llbracket - \rrbracket$ from $\beta$ to $\alpha$, i.e. $\llbracket - \rrbracket : X \to S^\omega$ is such that the following diagram commutes.
\begin{center}
\begin{tikzpicture}
  \matrix (m) [matrix of math nodes,row sep=3em,column sep=4em,minimum width=2em]
  {
     X & A^\omega \\
     A \times X & A \times A^\omega \\};
  \path[-stealth]
  (m-1-1) edge node [above] {$\llbracket - \rrbracket$} (m-1-2)
  (m-1-1) edge node [left] {$\beta$} (m-2-1)
  (m-1-2) edge node [right] {$\alpha$} (m-2-2)
  (m-2-1) edge node [below] {$id_A \times \llbracket - \rrbracket$} (m-2-2);
\end{tikzpicture}
\end{center}
Expressing $\beta = \langle o,t \rangle$ with $o : X \to A$ and $t : X \to X$, we have that $o(x) = \llbracket x \rrbracket(0)$ and $\llbracket t(x) \rrbracket(n) = \llbracket x \rrbracket (n+1)$ so that a direct formula is $\llbracket x \rrbracket (n) = o(t^n (x))$. This fits with the idea of an observable behaviour : what $\llbracket - \rrbracket$ tells is exactly the information we get if we look at the repeated output of $\beta$ over a state. Actually, the object $S^\omega$ is the set of streams over $S$, i.e., of $S$-valued sequences. The coalgebra $\beta$ generates over time the data contained in a stream, step by step.

\subsubsection*{Bisimulation (up-to)}
In the framework of coalgebras, the notion of \il{bisimulation} is a tool which provides a family of proof techniques (see \cite{Rot15}). There is a general definition of bisimulations in any category using diagrams. Since this paper is only using bisimulations for discrete systems, the definition given here is specific to the category $\gr{Sets}$. It is equivalent to the general definition \cite{Jacobs16}.
\begin{defi}[Bisimulation]\label{bisimulation}
Let $F : \gr{Sets} \to \gr{Sets}$ be an endofunctor, let $X$ be a set and let $R \subseteq X \times X$ be a relation. The \emph{relation lifting} of $R$ by $F$ is defined by
\[ Rel(F)(R) = \{ (b,c) \in FX \times FX \mid \exists d \in FR, b = F\pi_1(d) \text{ and } c = F\pi_2(d) \} \]
Given a coalgebra $\alpha : X \to FX$, let $b_\alpha : \mathcal{P} (X \times X) \to \mathcal{P}(X \times X)$ be defined by $b_\alpha(R) = (\alpha \times \alpha)^{-1}(Rel(F)(R))$. The relation $R$ is called a bisimulation on $\alpha$ if $R \subseteq b_\alpha(R)$.
\end{defi}
The greatest bisimulation on a given coalgebra $(X,\alpha)$ is called bisimilarity and denoted by $\sim$. The following statement expresses the principal interest of using bisimulations: bisimilarity implies behavioural equivalence.
\begin{lem} \label{bisim}
Let $(X,\alpha)$ be an $F$-coalgebra on $\textbf{Sets}$. Assuming there exists a final $F$-coalgebra $(Z,\zeta)$, we have
\[ \forall x,y \in X, x \sim y \Rightarrow \llbracket x \rrbracket = \llbracket y \rrbracket \]
\end{lem}
\begin{proof}
Let $R \subseteq X \times X$ be a bisimulation such that $(x,y) \in R$. For every $(x',y') \in X \times X$, let $d \in FR$ such that $\alpha(x') = F\pi_1(d)$ and $\alpha(y') = F\pi_2(d)$ and set $\gamma(x',y') = d$. The pair $(R,\gamma)$ is an $F$-coalgebra. Note that $\pi_1$ is a coalgebra morphism from $(R,\gamma)$ to $(X,\alpha)$ because $(F\pi_1 \circ \gamma)(x',y') = F\pi_1(d) = \alpha(x') = (\alpha \circ \pi_1)(x',y')$. The same is true for $\pi_2$. Thus $\llbracket - \rrbracket \circ \pi_1$ and $\llbracket - \rrbracket \circ \pi_2$ are coalgebra morphisms from $(R,\gamma)$ into the final coalgebra $(Z,\zeta)$, hence they are equal. Since $(x,y) \in R$ this yields $\llbracket x \rrbracket = \llbracket y \rrbracket$.
\end{proof}
\noindent It will be seen later that bisimulations incorporate a huge amount of elements. This is the reason why a slightly more refined notion is needed, namely bisimulation up-to.
\begin{defi} Let $(X,\alpha)$ be an $F$-coalgebra on $\gr{Sets}$ and $g : \mathcal{P} (X \times X) \to \mathcal{P} (X \times X)$ be a function. A relation $R \subseteq X \times X$ is a \emph{bisimulation up-to $g$} if $R \subseteq b_\alpha(g(R))$.
\end{defi}
\noindent For concrete examples of bisimulation (up-to), see the following part about deterministic and non-deterministic automata.

\subsection{Automata}
We present here the coalgebraic view of automata in the category $\gr{Sets}$, for a given alphabet $A$. We begin with Moore automata, which are a slight generalization of deterministic automata. Then we take a look at non-deterministic automata. The notion of bisimulation (up-to) is instantiated to deterministic and non-deterministic automata in order to present algorithms $\texttt{HK}$ and $\texttt{HKC}$ which are useful in section $2$, and a few examples are given.
\subsubsection*{Moore and deterministic automata}
Let $B$ be a set. Define the \il{machine functor} $F_B$ by $F_B X = B \times X^A$ and $Ff = id_B \times f^A$. An $F_B$-coalgebra models a Moore automaton with output in $B$. Let $\beta$ be an $F_B$-coalgebra. We often denote such a coalgebra by $\beta = \langle o , a \mapsto t_a \rangle$. Each state is mapped to a value via the output function $o : X \to B$ and, for every letter $a \in A$, to  a (unique) other state via $t_a : X \to X$. Note that automata, modeled coalgebraically in this way, do not have a notion of initial state. For such automata, it is useful to have a generalized notion of language. From now, a language will be a function $L : A^* \to B$. The language derivative is still defined by $L_a(w) = L(aw)$.\\\\
\gr{Remark.} The notation $2$ stands for the set $\{ 0,1 \}$. An $F_2$-coalgebra models a \textit{deterministic automaton}. For $x \in X$, the output is $1$ when the state is terminating and the output is $0$ when it is not. The category $\textbf{Coalg}(F_2)$ will be denoted by $\textbf{DA}$.\\\\
The following proposition underlines the main interest of using the machine functor: it has a final coalgebra, consisting of languages.
\begin{prop} \label{final coalgebra}
There exists a final $F_B$-coalgebra $(\Omega,\omega)$ where $\Omega = B^{A^*}$ and $\omega : B^{A^*} \to F_B B^{A^*}$ is defined by $\omega(L) = \langle L(\varepsilon), a \to L_a \rangle$.
\end{prop}
\begin{proof}
Let $\beta = \langle o, a \mapsto t_a \rangle : X \to FX$. We must prove that there exist a unique function $\varphi : X \to B^{A^*}$ such that the following diagram commutes.
\begin{center}
\begin{tikzpicture}
  \matrix (m) [matrix of math nodes,row sep=3em,column sep=4em,minimum width=2em]
  {
     X & \Omega \\
     FX & F\Omega \\};
  \path[-stealth]
  (m-1-1) edge node [above] {$\varphi$} (m-1-2)
  (m-1-1) edge node [left] {$\beta$} (m-2-1)
  (m-1-2) edge node [right] {$\omega$} (m-2-2)
  (m-2-1) edge node [below] {$F\varphi$} (m-2-2);
\end{tikzpicture}
\end{center}
This means $F\varphi \circ \beta = \omega \circ \varphi$ i.e. $ \langle o(x), a \mapsto (\varphi \circ t_a)(x) \rangle = \langle \varphi(x)(\varepsilon) , a \mapsto \varphi(x)_a \rangle $. Thus, we define inductively the function $\varphi$ by
\begin{align*}  \varphi(x)(\varepsilon) = o(x)  & & \varphi(x)(aw) = \varphi(t_a(x))(w) \end{align*}
By construction it makes the diagram commute. Furthermore, these two last equations are needed for commutation, hence such a morphism is unique.
\end{proof}
\noindent The language accepted by a state $x \in X$ can thus be defined using the final morphism by $\llbracket x \rrbracket (\varepsilon) = o(x)$ and $\llbracket x \rrbracket(aw) = \llbracket t_a(x) \rrbracket(w)$.
\subsubsection*{Bisimulations (up-to) with the functor $F_B$}
It is important to notice that the notion of bisimulation for $F_B$-coalgebras is very well-behaved. First, an easy computation shows that for any $R \subseteq X \times X$, there is the simple expression $b_\beta(R) = \{ (x,y) \in X \times X \mid o(x) = o(y) \land \forall a \in A, (t_a(x),t_a(y)) \in R \}$. Hence, a relation $R$ is a bisimulation iff for all $(x,y)\in R$, $o(x) = o(y)$ and for all $a\in A$, $(t_a(x),t_a(y))\in R$. Moreover, bisimilarity and trace equivalence coincide, as shown by the following lemma combined with Lemma \ref{bisim}.
\begin{lem} \label{bisimeq}
Let $\beta : X \to F_B X$ be an $F_B$-coalgebra. Then
\[ \forall x,y \in X, \llbracket x \rrbracket = \llbracket y \rrbracket \Rightarrow x \sim y \]
\end{lem}
\begin{proof}
Define by induction $t_\varepsilon = id_X$ and $t_{wa} = t_a \circ t_w$. Set $R = \{ ( t_w(x),t_w(y) ) \mid w \in A^* \}$ and show that $R$ is a bisimulation containing $(x,y)$. This last point is obvious because $(x,y) = (t_\varepsilon (x), t_\varepsilon (y))$. Let $(x',y') \in R$ and $w \in A^*$ such that $x' = t_w(x)$ and $y' = t_w(y)$. Then $o(x') = \llbracket t_w(x) \rrbracket (\varepsilon) = \llbracket x \rrbracket (w) = \llbracket y \rrbracket (w) = \llbracket t_w(y) \rrbracket (\varepsilon) = o(y')$. Let $a \in A$, then $(t_a(x'),t_a(y')) = (t_{wa}(x),t_{wa}(y)) \in R$, so that $R$ is a bisimulation. Consequently, $x \sim y$.
\end{proof}
\noindent Let $x,y \in X$. The following algorithm $\texttt{Naive}(x,y)$ tries to compute the smallest bisimulation that contains $(x,y)$, this is, $R = \{ (t_w(x), t_w(y)) \mid w \in A^* \}$ as in Lemma \ref{bisimeq}. If $R$ is not a bisimulation, the algorithm will stop at some point. If $R$ exists and is infinite (which requires $X$ to be infinite), then $\texttt{Naive}(x,y)$ never stops.

\begin{center}
$\texttt{Naive}(x,y)$
\end{center}
\texttt{(1) $\textit{R} := \emptyset$; $\textit{todo} := \emptyset$\\
(2) insert $(x,y)$ into $\textit{todo}$\\
(3) \color{purple} while \color{black} $\textit{todo}$ is not empty \color{purple}do\color{black}\\
\indent (3.1) extract $(x',y')$ from $\textit{todo}$\\
\indent (3.2) \color{purple}if \color{black} $(x',y')\in \textit{R}$ \color{purple}then \color{brown} continue \color{black}\\
\indent (3.3) \color{purple}if \color{black} $o(x') \neq o(y')$ \color{purple}then \color{brown} return \color{blue}\textit{false} \color{black}\\
\indent (3.4) \color{purple}for all \color{black}$a\in A$, insert $(t_a(x'), t_a(y'))$ into $\textit{todo}$\\
\indent (3.5) insert $(x',y')$ into $\textit{R}$\\
(4) \color{brown}return \color{blue} \textit{true}}
\\\\
Bisimulation is fine but, as we will see, bisimulation up-to $g$ is better - provided $g$ is \il{compatible} with $b_\beta$. Let $g : \mathcal{P}(X \times X) \rightarrow \mathcal{P}(X \times X)$. We say that $g$ is compatible with $b_\beta$ if it is monotone and for all $R,R' \subseteq Y \times Y$, $R \subseteq b_\beta(R') \Rightarrow g(R) \subseteq b_\beta(g(R'))$. Functions that are compatible with $b_\beta$ have the pleasant following property.
\begin{prop}[\cite{Bonchi13}]\label{contained}
For all $g : X \times X \rightarrow X \times X$ compatible with $b_\beta$, any bisimulation up-to $g$ is contained into a bisimulation.
\end{prop}
\begin{proof}
Let $R \subseteq X \times X$ such that $R \subseteq b_\beta(g(R))$. By compatibility and a simple induction, $g^n(R) \subseteq b_\beta(g^{n+1}(R))$ so that $\bigcup_{n\in \mathbb{N}} g^n(R) \subseteq b_\beta(\bigcup_{n\in \mathbb{N}} g^n(R))$ is a bisimulation that contains $R$.
\end{proof}
\noindent Hence, if $g$ is compatible with $b_\beta$ and $x,y$ are related by a bisimulation up-to $g$, then $x \sim y$.
\begin{lem} \label{comp1}
The following functions are compatible with $b_\beta$ :
\begin{itemize}
\item $r : R \mapsto \{ (x,x) \mid x \in X \}$
\item $s : R \mapsto \{ (y,x) \mid (x,y) \in R \}$
\item $t : R \mapsto \{ (x,z) \mid \exists y \in X, (x,y) \in R \text{ and } (y,z) \in R \}$
\item $id : R \mapsto R$
\item $f \circ g$ for all $f,g$ compatible with $b_\beta$
\item $\bigcup_{i \in I} f_i$ for all $(f_i)_{i\in I}$ compatible with $b_\beta$
\item $f^\omega = \bigcup_{n\in \mathbb{N}} f^n$ for all $f$ compatible with $b_\beta$
\end{itemize}
\end{lem}
\noindent The last two points of the latter lemma hold in general for any compatible function. This is actually the reason why compatible functions were first introduced.
\subsubsection*{Up-to equivalence and $\texttt{HK}$}
The $\texttt{Naive}$ algorithm can be improved using up-to techniques. The principle of the Hopcroft-Karp algorithm ($\texttt{HK}$) is to reason up-to equivalence. This can significantly speed up the \texttt{Naive} algorithm in many cases.
For a relation $R \subseteq X \times X$, the equivalence closure $e(R)$ is the least equivalence relation that contains $R$, i.e., that satisfies
$$
\infer{(x,y) \in e(R)}{(x,y) \in R}
\qquad
\infer{(x,x) \in e(R)}{}
\qquad
\infer{(y,x) \in e(R)}{(x,y) \in e(R)}
\qquad
\infer{(x,z) \in e(R)}{(x,y) \in e(R) & (y,z) \in e(R)}
$$
The function $e$ is compatible with $b_\beta$ because $e = (id \cup r \cup s \cup t)^\omega$. Now replace line $\texttt{(3.2)}$ in $\texttt{Naive}(x,y)$ by \\\\
\texttt{\indent (3.2) \color{purple}if \color{black} $(x',y')\in e(\textit{R})$ \color{purple}then \color{brown} continue \color{black}} \\\\
This new algorithm, $\texttt{HK}(x,y)$, will stop faster than $\texttt{Naive}(x,y)$ because it returns \texttt{true} as soon as it has built a bisimulation up-to $e$. Here is an example in the case of deterministic automata ($B = 2$, the ouput is $1$ iff the state is double-circled).
\begin{center}
\begin{tikzpicture}[->,>=stealth',shorten >=1pt,auto,node distance=2.8cm,
                    semithick]
  \tikzstyle{every state}=[fill=none,draw=black,text=black]

  \node[state](X1)   						{$x_1$};
  \node[state](Y1) [below of=X1]			{$y_1$};
  \node[state,accepting](X2) [right of=X1]	{$x_2$};
  \node[state,accepting](Y2) [right of=Y1]	{$y_2$};
  \path (X1) edge node {a} (X2)
		(Y1) edge node {a} (Y2)
		(X2) edge node [above] {b} (Y1)
		(Y2) edge node [below] {b} (X1);
  
\end{tikzpicture}
\end{center}
In this case, applying $\texttt{Naive}(x_1,y_1)$ constructs the smallest bisimulation in $4$ steps : $R = \{ (x_1,y_1),(x_2,y_2),(y_1,x_1),(y_2,x_2) \}$ and then returns \texttt{true}, whereas $\texttt{HK}(x_1,y_1)$ returns \texttt{true} as soon as $R' = \{ (x_1,y_1),(x_2,y_2) \}$ because $R'$ is a bisimulation up-to $e$.

\subsubsection*{Non-deterministic automata}
A non-deterministic automaton is a coalgebra for the composite functor $F_2 \mathcal{P}$. Let $\alpha$ be a non-deterministic automaton, then we can write it as $\alpha = \langle o , t_a \rangle : X \to 2 \times \mathcal{P}(X)^A$. The output $o$ still models termination. The difference with deterministic automata is the type of $t_a : X \to \mathcal{P}(X)$. The category $\gr{Coalg}(F_2\mathcal{P})$ will be denoted by $\gr{NDA}$. It does not have a final object, but given $\alpha : X \to 2 \times \mathcal{P}(X)^A$, the language of a state $x \in X$ is denoted by $\llbracket x \rrbracket_{\gr{NDA}}$ and can intuitively be defined by\begin{align*}
& \llbracket x \rrbracket_{\gr{NDA}} (\varepsilon) = o(x) & & \llbracket x \rrbracket_{\gr{NDA}} (aw) = \max_{y\in t_a(x)} \llbracket y \rrbracket_{\gr{NDA}} (w)
\end{align*}
Bisimilarity and trace equivalence with respect to this semantics do not coincide. One can prove that $R \subseteq X \times X$ is a bisimulation iff for all $(x,y) \in R$,
\begin{align*}
   & o(x) = o(y)
\\ & \forall x' \in t_a(x),  \exists y' \in X,  (x',y') \in R
\\ & \forall y' \in t_a(x),  \exists x' \in X,  (x',y') \in R
\end{align*}
Consider the following example :
\begin{center}
\begin{tikzpicture}[scale = 0.8, transform shape,->,>=stealth',shorten >=1pt,auto,node distance=2.8cm,
                    semithick]
  \tikzstyle{every state}=[fill=none,draw=black,text=black]

  \node[state](X)                    {$x$};
  \node[state](Y) 		[below left of=X] {$y$};
  \node[state,accepting](Z)		[below of=Y] {$z$};
  \node[state](Y')	[below right of=X] {$y'$};
  \node[state,accepting](Z') 	[below of=Y'] {$z'$};
  \path (X) edge [left] node {a} (Y)
		(Y) edge  [left] node {b} (Z)
		(X) edge [right] node {a} (Y')
 		(Y') edge [right] node {c} (Z');
 		
  \node (T) [right of=X] {};
  \node (U) [right of=T] {};	
  \node[state](X1)  [right of=U]     {$u$};
  \node[state](Y1) 		[below of=X1] {$v$};
  \node[state,accepting](Z1)		[below left of=Y1] {$w$};
  \node[state,accepting](Z1') 	[below right of=Y1] {$w'$};
  \path (X1) edge [left] node {a} (Y1)
		(Y1) edge  [left] node {b} (Z1)
 		(Y1) edge [right] node {c} (Z1');
  
\end{tikzpicture}
\end{center}
Here $\llbracket x \rrbracket_{\gr{NDA}} = \{ab,ac\} = \llbracket u \rrbracket_{\gr{NDA}}$ but $x$ and $u$ are not bisimilar. Indeed, if $x\sim u$ then $y \sim v$ but this is impossible because $y$ has no arrow labeled with $c$, whereas $v$ has one.\\\\
We recall the well-known power set construction : let $\langle o, a \mapsto t_a \rangle : X \rightarrow 2 \times \mathcal{P}(X)^A$. We define $\langle o^\#, a \mapsto t_a^\# \rangle : \mathcal{P}(X) \rightarrow 2\times \mathcal{P}(X)^A$ by $o^\#(U) = \max_{x\in U} o(x)$ and $t_a^\# (U) = \bigcup_{x \in U} t_a(x)$. Note that this is a deterministic automaton. Let $\llbracket - \rrbracket_{\gr{DA}}$ be the final morphism from $\langle o^\#, a \mapsto t_a^\# \rangle$ to the final object in $\gr{DA}$. The determinized automaton recognizes the same language as the first one, in the sense that for all $x\in X$, $\llbracket x \rrbracket_{\textbf{NDA}} = \llbracket \{x\} \rrbracket_{\textbf{DA}}$.
\subsubsection*{Up-to congruence and $\texttt{HKC}$}
Let $\alpha$ be a non-deterministic automaton. In order to check if $\llbracket x \rrbracket_{\textbf{NDA}} = \llbracket y \rrbracket_{\textbf{NDA}}$, it is possible to first compute the determinized automata $\alpha^\#$ and then check if $\llbracket \{ x \} \rrbracket_{\gr{DA}} = \llbracket \{ y \} \rrbracket_{\gr{DA}}$. According to Lemmas \ref{bisim} and \ref{bisimeq}, it is equivalent to check if $\{ x \} \sim \{ y \}$ in the determinized automata. To this purpose, one can use bisimulation or possibly bisimulation up-to equivalence, but an even better option is to exploit the determinized structure of $\alpha^\#$ with bisimulation up-to congruence. This is what the authors of \cite{Bonchi13} do, as an improvement of Hopcroft and Karp's algorithm for determinized automata. In section $2$, we will draw inspiration from $\texttt{HKC}$ to propose an algorithm for the trace semantics of PTS.\\\\
Let $R \subseteq \mathcal{P}(X) \times \mathcal{P}(X)$. Its congruence closure $c(R)$ is the least congruence relation that contains $R$, i.e., that satisfies
$$
\infer{(U,V) \in c(R)}{(U,V) \in R}
\qquad
\infer{(U,U) \in c(R)}{}
\qquad
\infer{(V,U) \in c(R)}{(U,V) \in c(R)}
\qquad
\infer{(U,V) \in c(R)}{(V,W) \in c(R) & (U,W) \in c(R)}
$$
$$
\infer{(U \cup U' , V \cup V') \in c(R)}{(U,V) \in c(R) & (U',V') \in c(R)}
$$
The function $c$ is shown to be compatible with $b_{\alpha^\#}$ in \cite{Bonchi13}. Thus, if $\{x\}$ and $\{y\}$ are related by a bisimulation up-to congruence, they are bisimilar and this yields $\llbracket x \rrbracket_{\gr{NDA}} = \llbracket y \rrbracket_{\gr{NDA}}$. Hence the following algorithm computes the trace equivalence of $x$ and $y$.
\begin{center}
$\texttt{HKC}(x,y)$
\end{center}
\texttt{(1) $\textit{R} := \emptyset$; $\textit{todo} := \emptyset$\\
(2) insert $(\{x\},\{y\})$ into $\textit{todo}$\\
(3) \color{purple} while \color{black} $\textit{todo}$ is not empty \color{purple}do\color{black}\\
\indent (3.1) extract $(U,V)$ from $\textit{todo}$\\
\indent (3.2) \color{purple}if \color{black} $(U,V)\in c(\textit{R})$ \color{purple}then \color{brown} continue \color{black}\\
\indent (3.3) \color{purple}if \color{black} $o^\#(U) \neq o^\#(V)$ \color{purple}then \color{brown} return \color{blue}\textit{false} \color{black}\\
\indent (3.4) \color{purple}for all \color{black}$a\in A$, insert $(t_a^\#(U), t_a^\#(V))$ into $\textit{todo}$\\
\indent (3.5) insert $(U,V)$ into $\textit{R}$\\
(4) \color{brown}return \color{blue} \textit{true}}
\\\\
The following example taken from \cite{Bonchi12} shows how $\texttt{HKC}$ quickens the computations. The alphabet is $A = \{ a \}$. The initial non-deterministic automaton is:
\begin{center}
\begin{tikzpicture}[scale = 0.8, transform shape,->,>=stealth',shorten >=1pt,auto,node distance=2.8cm,
                    semithick]
  \tikzstyle{every state}=[fill=none,draw=black,text=black]

  \node[state,accepting](X)                    {$x$};
  \node[state,accepting](Y) 	[right of=X] {$y$};
  \node[state](Z)		[right of=Y] {$z$};
  \node[state,accepting](U) [below of=X] {$u$};
  \path (X) edge [bend left] node {a} (Y)
  		(Y) edge node {a} (X)
  		(X) edge [bend right] node {a} (Z)
  		(Z) edge node {a} (Y)
  		(U) edge [loop right] node {a} (U);

\end{tikzpicture}
\end{center}
The (interesting part of the) determinized automaton is:
\begin{center}
\begin{tikzpicture}[scale = 0.8, transform shape,->,>=stealth',shorten >=1pt,auto,node distance=2.8cm,
                    semithick]
  \tikzstyle{every state}=[fill=none,draw=black,text=black]

  \node[state,accepting](X1) [right of=Z] {$\{x\}$};
  \node[state,accepting](X2) [right of=X1] {$\{y,z\}$};
  \node[state,accepting](X3) [right of=X2] {$\{x,y\}$};
  \node[state,accepting](X4) [right of=X3] {$\{x,y,z\}$};
  \node[state,accepting](X5) [below of=X1] {$\{u\}$};

  \path (X1) edge node {a} (X2)
  		(X2) edge node {a} (X3)
  		(X3) edge node {a} (X4)
  		(X5) edge [loop right] node {a} (X5);

 \draw[-,dashed] (X1) -- (X5) -- (X2);
 \draw[-,dotted] (X3) -- (X5) -- (X4);
  
\end{tikzpicture}
\end{center}
The bisimulation computed by $\texttt{Naive}(\{x\},\{u\})$ or $\texttt{HK}(\{x\},\{u\})$ is represented by dotted + dashed lines. The bisimulation up-to congruence computed by $\texttt{HKC}(x,y)$ consists only in the dashed lines: after two steps, the algorithm $\texttt{HKC}$ returns \texttt{true}.

\section{Discrete case}
Define the endofunctor $L : \gr{Sets} \to \gr{Sets}$ by setting $LX = A \times X + 1$ and $Lf = id_A \times f + id_1$. Remember that $(P,\eta,\mu)$ is the probability distribution monad. We can now formally define PTS from the introduction as $PL$-coalgebras. For a PTS $\alpha : X \to PLX$, the (finite and infinite) trace semantics $\ll - \rr : X \to \mathcal{M}(A^\infty)$ can be defined by the following equations.
\begin{align*}
&\ll x \rr (\varepsilon A^\infty) = 1 & & \ll x \rr (\varepsilon) = \alpha(*) \\
&\ll x \rr (awA^\infty) = \sum_{y\in X} \alpha(x)(a,y) \cdot \ll y \rr (wA^\infty) &
&\ll x \rr (aw) = \sum_{y\in X} \alpha(x)(a,y) \cdot \ll y \rr (w)
\end{align*}
The aim of this section is to find back this semantics via a determinization construction and to provide an efficient algorithm that takes $x,y \in X$ and checks if $\ll x \rr = \ll y \rr$. This last part is challenging at first sight, because this means the algorithm checks the equality of two elements of $\M (A^\infty)$, which are functions on the uncountable space $\Sigma_{A^\infty}$. \\\\
To begin with, let us get used to this semantics and to measures over words by doing some computations for two different PTS. The first one is defined on the alphabet $A = \{a\}$.
\begin{center}
\begin{tikzpicture}[->,>=stealth',shorten >=1pt,auto,node distance=2.8cm,
                    semithick]
  \tikzstyle{every state}=[fill=none,draw=black,text=black]

  \node[state](X)                    {$x$};
  \node[state](Y) [right of=X] {$y$};
  \node[state,accepting](F) [right of=Y] {$*$};

  \path (Y) edge node {a,1/3} (X)
		(Y) edge [loop above] node {a,1/3} (Y)
		(Y) edge              node {1/3} (F)
 		(X) edge [loop above] node {a,1} (X);
  
\end{tikzpicture}
\end{center}
The language $\ll x \rr$ is very easy to compute for sets of words in $S_\infty$ by induction: for every finite word $w$, $\ll x \rr (w) = 0$ and $\ll x \rr (wA^\infty) = 1$. Hence, we have for example $\ll x \rr (a^\omega) = \lim_{n\to \infty} \ll x \rr (a^n A^\infty) = 1$ because the sequence of sets $(a^n A^\infty)_{n\in \N}$ is decreasing. Let us look at $\ll y \rr$. First see that $\ll y \rr (\varepsilon) = 1/3$ and $\ll y \rr (\varepsilon A^\infty) = 1$. Let $n \in \N \cup \{0\}$, then:
\begin{align*}
&\ll y \rr (a^{n+1}) = \frac{1}{3} \ll y \rr (a^n) + \frac{1}{3} \ll x \rr (a^n) = \frac{1}{3} \ll y \rr (a^n) \\
&\ll y \rr (a^{n+1}A^\infty) = \frac{1}{3} \ll y \rr (a^n A^\infty) + \frac{1}{3} \ll x \rr (a^n A^\infty) = \frac{1}{3} \ll y \rr (a^n A^\infty) + \frac{1}{3} 
\end{align*}
Hence $\ll y \rr (a^n) = 1/3^{n+1}$ and $\ll y \rr (a^n A^\infty) = (1 + 3^{-n})/2$. Using the same arguments as for $\ll x \rr$, we have $\ll y \rr (a^\omega) = \lim_{n\to \infty} (1+3^{-n})/2 = 1/2$. This is rather intuitive: the probability of performing $n$ loops in state $y$ and then getting lost forever in state $x$ is $1/3^{n+1}$. Summing them for $n\in \N \cup \{0\}$ gives $1/2$.
\\\\
The second PTS is defined on the alphabet $A = \{ \gr{0}, \gr{1}, \gr{2} \}$. 
\begin{center}
\begin{tikzpicture}[->,>=stealth',shorten >=1pt,auto,node distance=2.8cm,
                    semithick]
  \tikzstyle{every state}=[fill=none,draw=black,text=black]

  \node[state](X)                    {$x$};
  \node[state](Y) 		[right of=X] {$y$};
  \node[state,accepting](F) [right of=Y] {$*$};

  \path (X) edge [loop above] node {\textbf{0},1/3} (X)
		(X) edge [loop below] node {\textbf{2},1/3} (X)
		(X) edge              node {\textbf{1},1/3} (Y)
 		(Y) edge              node {1} (F);
  
\end{tikzpicture}
\end{center}
This automaton generates Cantor's space, in the sense that when generating an infinite word, the computation stops if and only if this word is not in Cantor's space (we recall this is the set of real numbers $r\in [0,1]$ such that there is no 1 in the base 3 expansion of $r$). More precisely, we have $\ll x \rr (w\textbf{1}) = \left(\frac{1}{3}\right)^{|w|+1}$ if there is no $\textbf{1}$ in $w$ and $\ll x \rr (w) = 0$ for other finite words. Furthermore $\ll x \rr (wA^\infty) = \left(\frac{1}{3}\right)^{|w|}$ if there is no \textbf{1} in $w$. We can guess this way that the measure of Cantor's space in $[0,1]$ is 0 :
\begin{align*}
\ll x \rr (A^\omega) &= \ll x \rr (A^\infty \setminus A^*) =  \ll x \rr (A^\infty) - \ll x \rr (A^*) = 1 - \ll x \rr \left(\bigcup_{w\in A^*} \{w\}\right) = 1 - \sum_{w\in A^*} \ll x \rr (w) \\ &= 1 - \sum_{w\in \{\textbf{0},\textbf{2}\}^*} \ll x \rr (w\textbf{1}) = 1 - \sum_{n\in\mathbb{N}} \sum_{w \in \{\textbf{0},\textbf{2}\}^*, |w|=n} \left( \frac{1}{3} \right)^{n+1} = 1 - \frac{1}{3} \sum_{n\in\mathbb{N}} 2^n \left( \frac{1}{3} \right)^n = 0
\end{align*}
\subsection{Trace semantics via determinization}
Starting from $\alpha : X \to PLX$, we will proceed in three steps in order to define its trace semantics morphism $\ll - \rr : X \to \M (A^\infty)$. 
\begin{enumerate}
\item[{\color{blue}(i)}] Translate $\alpha$ into a coalgebra for the more convenient composite functor $F_{\I \times \I} D$, obtaining an $\tilde{\alpha} : X \to F_{\I \times \I} DX$. Recall that $(D,\eta,\mu)$ is the sub-probability distribution monad. In the sequel $F_{\I \times \I}$ will be denoted by $F$.
\item[{\color{red}(ii)}] Determinize it, i.e., define an $\tilde{\alpha}^\# : DX \to FDX$ such that $\tilde{\alpha}^\# \circ \eta_X = \tilde{\alpha}$. Thus there is a final morphism $\varphi_{\tilde{\alpha}^\#} : DX \to \Omega$.
\item[(iii)] Factorize the final morphism to get a coalgebra morphism $DX \to \M (A^\infty)$, then precompose with $\eta_X$ to get the desired trace semantics $X \to \M (A^\infty)$.
\end{enumerate}
The whole construction is summed up in the following diagram. Here $\Omega = (\I \times \I)^{A^*}$ is the set of languages with two outputs in $\I$ in accordance with Proposition \ref{final coalgebra}.
\begin{center}
\begin{tikzpicture}
  \matrix (m) [matrix of math nodes,row sep=2em,column sep=4em,minimum width=2em]
  {
     X & D X & \M (A^\infty) & \Omega \\
     P L X & & & \\
     D L X & F D X & F \M (A^\infty) & F\Omega \\};
  \path[-stealth]
    (m-1-1) edge [color=blue] node [left] {$\alpha$} (m-2-1)
    (m-2-1) edge [color=blue] node [left] {$\iota_{LX}$} (m-3-1)
    (m-1-1) edge [color=red] node [below] {$\eta_X$} (m-1-2)
    (m-1-2) edge [color=black] node [below] {$\llbracket - \rrbracket$} (m-1-3)
    (m-1-3) edge [color=black] node [below] {$\varphi_{\Pi}$} (m-1-4)
    (m-1-2) edge [color=red] node [right] {$\tilde{\alpha}^\#$} (m-3-2)
    (m-1-3) edge [color=black] node [right] {$\Pi$} (m-3-3)
    (m-1-4) edge [color=red] node [right] {$\omega$} (m-3-4)
    (m-3-1) edge [color=blue] node [below] {$\e_X$} (m-3-2)
    (m-3-2) edge [color=black] node [above] {$F\llbracket - \rrbracket$} (m-3-3)
    (m-3-3) edge [color=black] node [above] {$F\varphi_{\Pi}$} (m-3-4)
    (m-1-2) edge [color=red, bend left] node [above] {$\varphi_{\tilde{\alpha}^\#}$} (m-1-4)
    (m-3-2) edge [color=red, bend right] node [below] {$F\varphi_{\tilde{\alpha}^\#}$} (m-3-4)
    (m-1-1) edge [color=blue] node [above right] {$\tilde{\alpha}$} (m-3-2)
    (m-1-1) edge [color=black, bend left] node [above] {$\ll - \rr$} (m-1-3)

    ;
\end{tikzpicture}
\end{center}
\subsubsection*{(i) Translation: from $\alpha$ to $\tilde{\alpha}$}
Let $\alpha : X \to PLX$ be a $PL$-coalgebra. The space $PLX$ does not enlighten the information provided by $\alpha$. It would be more convenient to have an explicit machine behaviour, i.e., make functor $F$ appear. The aim of the step $(i)$ is to formally change $PLX$ into $FDX$.\\\\
Each state $x$ gives rise to a probability distribution $\alpha(x)$ over $LX = A \times X + 1$. There are three important pieces of information in it. First, the information that the total mass is $1$. Second, the probability of termination, which is $\alpha(x)(*)$. Third, all the probabilities $\alpha(x)(a,y)$ that the system transitions to some $y \in X$ using some $a \in A$. Any faithful translation of $\alpha$ should keep in memory this information. We define the natural transformations $\iota : P \Rightarrow T$ and $\e : DL \Rightarrow FD$.
\begin{align*}
&\iota_Y : u \in PY \mapsto u \in DY \\
&\e_X : u \in DLX \mapsto \left\langle \sum_{z \in LX} u(z) , u(*) , a \mapsto [y \mapsto u(a,y)]\right\rangle \in FDX 
\end{align*}
The first one is from the inclusion $PY \subseteq DY$. The second one aims to clearly separate the useful information; it appears in \cite{Jacobs15} where it plays the role of the \il{extension natural transformation}. Now, take $\tilde{\alpha} = \langle \tilde{\alpha}_\1, \tilde{\alpha}_\0, a \mapsto t_a \rangle = \e_X \circ \iota_{LX} \circ \alpha$.
\begin{center}
\begin{tikzpicture}
  \matrix (m) [matrix of math nodes,row sep=2em,column sep=4em,minimum width=2em]
  { X & PLX & DLX & FDX \\};
  \path[-stealth]
    (m-1-1) edge node [below] {$\alpha$} (m-1-2)
    (m-1-2) edge node [below] {$\iota_{LX}$} (m-1-3)
    (m-1-3) edge node [below] {$\e_X$} (m-1-4)
    (m-1-1) edge [bend left] node [above] {$\tilde{\alpha}$} (m-1-4)
    ;
\end{tikzpicture}
\end{center}
An explicit expression is $\tilde{\alpha}(x) = \left\langle \underbrace{\sum_{z \in LX} \alpha(x)(z)}_{=1}, \alpha(x)(*), a \mapsto [y \mapsto \alpha(x)(a,y)] \right\rangle$. 

\subsubsection*{(ii) Determinization}
The $FD$-coalgebra $\tilde{\alpha}$ has now to be changed into an $F$-coalgebra in order to benefit from the final $F$-coalgebra semantics. This is the aim of step $(ii)$.\\\\
Let $\beta = \langle \beta_\1, \beta_\0, a \to t_a \rangle : X \to FDX$ be any $FD$-coalgebra. Then, there is an obvious way to determinize it into an $F$-coalgebra $\beta^\# = \langle \beta^\#_\1, \beta^\#_\0, a \to \tau_a \rangle : DX \to FDX$. Set for all $u \in DX$:
\begin{align*}
&\beta^\#_\1 (u) = \sum_{x \in X} u(x) \beta_\1 (x) \\
&\beta^\#_\0 (u) = \sum_{x \in X} u(x) \beta_\0 (x) \\
\forall a \in A \indent &\tau_a (u) = y \mapsto \sum_{x \in X} u(x) t_a(x)(y)
\end{align*}
This transformation is the same as in section $7.2$ of \cite{Jacobs15} if the first output is dropped. The state space of the automata is basically changed from $X$ to $DX$. In a sense, the function $t_a : X \to DX$ is thus "homogenized" into $\tau_a : DX \to DX$. As an $F$-coalgebra, $\beta^\#$ can be structurally seen as a Moore automaton. It takes $u \in DX$, outputs its total mass, outputs its termination mass, and every letter $a$ makes it transition to a new sub-probability distribution $v = \tau_a(u)$.\\\\
In this paragraph we explain the terminology "determinization". This refers to the classical power set construction, which consists in changing the state space $X$ of a non-deterministic automaton $\beta$ into $\mathcal{P}(X)$ to produce a \il{deterministic} automaton $\beta^\#$. This is actually an instance of a construction discussed in section $3$ (see Lemma \ref{deter}): the general determinization approach involves an Eilenberg-Moore category that we overlook for the moment. The key property is that the behaviour of $\beta$ starting with $x$ as initial state is the same as the behaviour of $\beta^\#$ starting with $\{ x \}$ as initial state.  Here we are actually doing the same thing, replacing the functor $\mathcal{P}$ with the functor $D$, and unit $x \mapsto \{ x \}$ with unit $x \mapsto \delta_{x,-}$. In terms of diagrams, what has to be checked to confer the title of determinization is that $\beta^\# \circ \eta_X = \beta$. This is the case: 
\begin{align*}
&(\beta^\#_\1 \circ \eta_X)(x)= \sum_{x'\in X} \delta_{x,x'} \beta_\1(x') = \beta_\1(x)\\
&(\beta^\#_\0 \circ \eta_X)(x)= \sum_{x'\in X} \delta_{x,x'} \beta_\0(x') = \beta_\0(x)\\
\forall a \in A \indent & (\tau_a \circ \eta_X)(x) = \left[y \mapsto \sum_{x'\in X} \delta_{x,x'} t_a(x')(y)\right] = [y \mapsto t_a(x)(y)] = t_a(x)
\end{align*}

\subsubsection*{(iii) Factorization of the final morphism}
Because $\tilde{\alpha}^\#$ is an $F$-coalgebra, there exists a unique coalgebra morphism $\varphi_{\tilde{\alpha}^\#}$ from $\tilde{\alpha}^\#$ to $\omega$. However, we want the semantics of $\tilde{\alpha}^\#$ to be a probability distribution over words, hence to live in $\M (A^\infty)$. The step $(iii)$ thus aims to factorize $\varphi_{\tilde{\alpha}^\#}$ using a new $F$-coalgebra $\Pi : \M (A^\infty) \to F \M (A^\infty)$ and a new coalgebra morphism $\llbracket - \rrbracket$ from $\tilde{\alpha}^\#$ to $\Pi$.
\begin{center}
\begin{tikzpicture}
  \matrix (m) [matrix of math nodes,row sep=2em,column sep=4em,minimum width=2em]
  { & & DX & \M (A^\infty) & \Omega & &  (*) \\};
  \path[-stealth]
    (m-1-3) edge node [below] {$\llbracket - \rrbracket$} (m-1-4)
    (m-1-4) edge node [below] {$\varphi_\Pi$} (m-1-5)
    (m-1-3) edge [bend left] node [above] {$\varphi_{\tilde{\alpha}^\#}$} (m-1-5)
    ;
\end{tikzpicture}
\end{center}
First, define $\Pi : \M (A^\infty) \to F \M (A^\infty)$ by
\[ \Pi(m) = \langle m(\varepsilon A^\infty), m(\varepsilon), a \mapsto m_a \rangle \]
The $F$-coalgebra $\Pi$ will play the role of a final object in a certain subcategory of $\gr{Coalg}(F)$. The unique coalgebra morphism from $\Pi$ to $\omega$ is denoted by $\varphi_\Pi$. It is injective (see Lemma \ref{injective}). The following lemma states in which cases the factorization is possible.
\begin{prop} \label{keybisd}
Let $\beta = \langle \beta_\1, \beta_\0, a \mapsto \tau_a \rangle : Y \to FY$ be an $F$-coalgebra. The two following conditions are equivalent:
\begin{enumerate}
\item[(i)] There exists an $F$-coalgebra morphism $\llbracket - \rrbracket$ from $\beta$ to $\Pi$.
\item[(ii)] The equation $\beta_\1 = \beta_\0 + \sum_{a\in A} \beta_\1 \circ \tau_a$ holds.
\end{enumerate}
In this case, this morphism is unique.
\end{prop}
\begin{proof} Refer to the proof of Theorem \ref{keybis} which is the same in the general case.
\end{proof}
\noindent \textbf{Remark.} If $(i)$ and $(ii)$ are true, then $\varphi_\Pi \circ \llbracket - \rrbracket$ is a coalgebra morphism from $\beta$ to $\omega$, hence $\varphi_\Pi \circ \llbracket - \rrbracket = \varphi_\beta$ by uniqueness. Using the injectivity of $\varphi_\Pi$, one can see that for any $u,v \in TX$, $\llbracket u \rrbracket = \llbracket v \rrbracket \Leftrightarrow \varphi_\beta(u) = \varphi_\beta(v)$. This is important because checking if $\varphi_\beta(u) = \varphi_\beta(v)$ is made easy by  the use bisimulation up-to techniques.\\\\
Let us check that $\tilde{\alpha}^\# : DX \to FDX$ satisfies $(ii)$ in Proposition \ref{keybisd}.
\begin{align*}
\tilde{\alpha}^\#_\1 (u) &= \sum_{x\in X} u(x) \tilde{\alpha}_\1(x) =  \sum_{x\in X} u(x) \sum_{z \in LX} \alpha(x)(z) \\ 
&= \sum_{x\in X} u(x) \alpha(x)(*) + \sum_{x \in X} \sum_{(a,y) \in A \times X} u(x) \alpha(x)(a,y) \\
&= \sum_{x\in X} u(x) \tilde{\alpha}_\0 (x) + \sum_{a\in A} \sum_{y\in X} \tau_a(u)(y) \\
&= \tilde{\alpha}^\#_\0(u) + \sum_{a\in A} (\tilde{\alpha}^\#_\1 \circ \tau_a) (u) & \text{\indent (because $\tilde{\alpha}_\1 (y) = 1$)}
\end{align*}
This achieves the proof of the following proposition.
\begin{prop}
The morphism $\varphi_{\tilde{\alpha}^\#}$ decomposes as a unique coalgebra morphism $\llbracket - \rrbracket$ from $\tilde{\alpha}^\#$ to $\Pi$ followed by an injective coalgebra morphism $\varphi_\Pi$ from $\Pi$ to $\omega$, as in $(*)$.
\end{prop}
\noindent We define $\ll - \rr = \llbracket - \rrbracket \circ \eta_X$ as our definitive semantics.
\subsubsection*{Link with the usual semantics}
We conclude this section by comparing the trace semantics via determinization we just obtained with the trace semantics of PTS defined previously.  
First express $\llbracket - \rrbracket_{|S_\infty}$ (see the proof of Theorem \ref{keybis}). For all $u \in DX$, $a \in A$, $w \in A^*$:
\begin{align*}
&\llbracket u \rrbracket (\varepsilon A^\infty) = \tilde{\alpha}^\#_\1 (u) &
&\llbracket u \rrbracket (\varepsilon) = \tilde{\alpha}^\#_\0 (u) \\
&\llbracket u \rrbracket (awA^\infty) = \llbracket \tau_a(u) \rrbracket (wA^\infty) &
&\llbracket u \rrbracket (aw) = \llbracket \tau_a(u) \rrbracket (w)
\end{align*}
The following lemma makes a rather intuitive link between $\llbracket - \rrbracket$ and $\ll - \rr$ and will help us to get back to the first semantics defined in the beginning of this section.
\begin{lem}
Let $u \in DX$. Then $\llbracket u \rrbracket = \sum_{x\in X} u(x) \ll x \rr$.
\end{lem} 
\begin{proof}
By induction on $w$ (for all $u$). First see that $\llbracket u \rrbracket (\varepsilon) = \tilde{\alpha}_\0^\# (u) = \sum_{x\in X} u(x) \tilde{\alpha}_\0 (x)= \sum_{x\in X} u(x) \ll x \rr (\varepsilon)$. Then $\llbracket u \rrbracket (aw) = \llbracket \tau_a(u) \rrbracket (w) = \sum_{y\in X} \tau_a(u)(y) \ll y \rr(w)$ by induction hypothesis. So $\llbracket u \rrbracket (aw) = \sum_{x,y\in X} u(x) t_a(x)(y) \ll y \rr (w) = \sum_{x \in X} u(x) \llbracket t_a(x) \rrbracket (w)$ again by induction hypothesis. But $\llbracket t_a(x) \rrbracket (w) = \llbracket (\tau_a \circ \eta_X) (x) \rrbracket (w) = \ll x \rr (aw)$ because $\llbracket - \rrbracket$ is a coalgebra morphism, so $\llbracket u \rrbracket (aw) = \sum_{x\in X} u(x) \ll x \rr (aw)$. The proof is the same for $\llbracket u \rrbracket (wA^\infty) = \sum_{x\in X} u(x) \ll x \rr (wA^\infty)$. Then, both measures coincide on $S_\infty$ so $\llbracket u \rrbracket = \sum_{x\in X} u(x) \ll x \rr$ according to Theorem \ref{key}.
\end{proof}

\noindent Using this lemma we have the following expressions.
\begin{align*}
&\ll x \rr (\varepsilon A^\infty) = (\tilde{\alpha}_\1^\# \circ \eta_X)(x) = \tilde{\alpha}_\1(x) = \sum_{z\in X}\alpha(x)(z) = 1 \\
&\ll x \rr (\varepsilon) = (\tilde{\alpha}_\0^\# \circ \eta_X)(x) = \tilde{\alpha}_\0(x) = \alpha(*) \\
&\ll x \rr (awA^\infty) = \llbracket (\tau_a \circ \eta_X)(x) \rrbracket (wA^\infty) = \llbracket t_a(x) \rrbracket (wA^\infty) =  \sum_{y \in X} \alpha(x)(a,y) \cdot \ll y \rr (wA^\infty) \\
&\ll x \rr (aw) = \llbracket (\tau_a \circ \eta_X)(x) \rrbracket (w) = \llbracket t_a(x) \rrbracket (w) = \sum_{y\in X} \alpha(x)(a,y) \cdot \ll y \rr (w)
\end{align*}
These are the same expressions as in the beginning of this section. Hence:
\begin{prop} The two trace semantics denoted by $\ll - \rr$ coincide. \end{prop}
\subsection{Algorithm}
Let $\alpha : X \to PLX$ and $x,y \in X$. The aim of this paragraph is to give an algorithm that takes $x,y \in X$ and tells whether $\ll x \rr = \ll y \rr$ or not. It will directly profit from the determinization construction performed in the last section. Our algorithm is greatly inspired by $\texttt{HKC}$, and more specifically, by its generalization to weighted automata given in \cite{Bonchi17}. Indeed, the very same bisimulation up-to congruence techniques can leverage the determinized structure of $\tilde{\alpha}^\#$.\\\\
Note that $\tilde{\alpha}^\#$ has type $DX \to \I \times \I \times (DX)^A$. Let $\varphi_{\tilde{\alpha}^\#}$ be the final morphism $DX \to (\I \times \I)^{A^*}$. The $F_{\I \times \I}$-coalgebra $\tilde{\alpha}^\#$ can be easily injected into the "bigger" $F_{\R \times \R}$-coalgebra $\beta : \R_\omega^X \to \R \times \R \times (\R_\omega^X)^A$ where $\R_\omega^X$ is the set of finitely supported functions $X \to \R$. Thus we get a vector space structure (see \cite{Bonchi12}). Indeed, just define for every $u \in \R_\omega^X$:
\[ \beta(u) = \left\langle \sum_{x \in X, a \in A, y \in X} u(x) \alpha(x)(a,y) , \sum_{x\in X} u(x) \alpha(x)(*) , a \mapsto \left[ y \mapsto \sum_{x\in X} u(x) \alpha(x)(a,y) \right] \right\rangle \]
We study bisimulations in the sense of $\beta = \langle \beta_\1, \beta_\0, a \mapsto \tau_a \rangle : \R_\omega^X \to F\R_\omega^X$. Note that $\beta_\1,\beta_\0$ and each $\tau_a$ are linear functions. Let $\varphi_\beta$ be the final morphism $\R_\omega^X \to (\R \times \R)^{A^*}$. The following lemma shows that changing $\tilde{\alpha}^\#$ into $\beta$ does not affect trace equivalence. 
\begin{lem}
Let $i$ be the injection $DX \to \R_\omega^X$ and $j$ be the injection $(\I \times \I)^{A^*} \to (\R \times \R)^{A^*}$. Then $j \circ \varphi_{\alpha^\#} = \varphi_{\beta} \circ i$.
\end{lem}
\noindent Now we omit writing $i$ and $j$. The following corollary allows us to work with bisimulation up-to congruence on $\R_\omega^X$, as defined below.
\begin{cor} For every $x,y \in X$, $\ll x \rr = \ll y \rr$ iff $\varphi_\beta (\delta_x) = \varphi_\beta(\delta_y)$. \end{cor}
\noindent Let $R \subseteq \mathbb{R}_\omega^X \times \mathbb{R}_\omega^X$. Its congruence closure $c(R)$ is the least congruence relation that contains $R$, i.e., that satisfies
$$
\infer{(u,v) \in c(R)}{(u,v) \in R}
\qquad
\infer{(u,u) \in c(R)}{}
\qquad
\infer{(v,u) \in c(R)}{(u,v) \in c(R)}
\qquad
\infer{(u,w) \in c(R)}{(u,v) \in c(R) & (v,w) \in c(R)}
$$
$$
\infer[(\lambda \in \mathbb{R})]{(\lambda u,\lambda v) \in c(R)}{(u,v) \in c(R)} 
\qquad 
\infer{(u+u',v+v') \in c(R)}{(u,u') \in c(R) & (v,v') \in c(R)}
$$

\begin{lem}\label{comp}
The functions $a : R \mapsto \{z + z' \mid z \in R, z' \in R\}$ and $m_\lambda : R \mapsto \{\lambda z \mid z \in R \}$ are compatible with $b_\beta$ for all $\lambda \in \mathbb{R}$.
\end{lem}
\begin{proof}
Let $\lambda \in\mathbb{R}$. If $R \subseteq b_\beta (R')$ and we take $(u',v')\in m_\lambda(R)$ then there exists $(u,v)\in R$ such that $(u',v')=(\lambda u ,\lambda v)$. We have $(u,v)\in b_\beta (R')$ so  $o(u') = o(\lambda u) = \lambda o(u) = \lambda o(v) = o(\lambda v) = o(v')$. Given $a\in A$ we have $(\tau_a(u),\tau_a(v))\in R'$ so $(\tau_a(u'),\tau_a(v')) = (\lambda \tau_a(u),\lambda \tau_a(v)) \in m_\lambda(R')$. Thus $m_\lambda(R)\subseteq b_\beta (m_\lambda(R'))$.\\
Let $(u'',v'') \in a(R)$, so $(u'',v'') = (u+u',v+v')$ where $(u,v) \in R$ and $(u',v') \in R$. So $o(u'') = o(u) + o(u') = o(v) + o(v') = o(v'')$ and for all $a \in A$, $(\tau_a(u''),\tau_a(v'')) = (\tau_a(u) + \tau_a(u'),\tau_a(v) + \tau_a(v')) \in a(R')$. Thus $a(R) \subseteq b_\beta (a(R'))$.
\end{proof}
\noindent According to Lemmas  \ref{comp1} and \ref{comp}, the function $c : \mathcal{P}(\R_\omega^X \times \R_\omega^X) \to \mathcal{P}(\R_\omega^X \times \R_\omega^X)$ is compatible with $b_\beta$, because $c = \left(id \cup r \cup s \cup t \cup a \cup \bigcup_{\lambda \in \mathbb{R}} m_\lambda\right)^\omega$. Thus, if $u,v \in \R_\omega^X$ are related by a bisimulation up-to congruence, they are bisimilar, and this yields $\varphi_\beta(u) = \varphi_\beta(v)$. In our case, when $u = \delta_x$ and $v = \delta_y$ for some $x,y \in X$, this further yields $\ll x \rr = \ll y \rr$. The following algorithm $\texttt{HKC}^\infty$ computes the smallest bisimulation that relates $x$ and $y$, hence it computes whether $\ll x \rr = \ll y \rr$ or not.

\begin{center} $\texttt{HKC}^\infty(x,y)$ \end{center}
\texttt{(1) $\textit{R} := \emptyset$; $\textit{todo} := \emptyset$\\
(2) insert $(\delta_x,\delta_y)$ into $\textit{todo}$\\
(3) \color{purple} while \color{black} $\textit{todo}$ is not empty \color{purple}do\color{black}\\
\indent (3.1) extract $(u,v)$ from $\textit{todo}$\\
\indent (3.2) \color{purple}if \color{black} $(u,v)\in c(\textit{R})$ \color{purple}then \color{brown} continue \color{black}\\
\indent (3.3) \color{purple}if \color{black} $\beta_\1(u) \neq \beta_\1(v)$ \color{purple}then \color{brown} return \color{blue}\textit{false} \color{black}\\
\indent (3.3') \color{purple}if \color{black} $\beta_\0(u) \neq \beta_\0(v)$ \color{purple}then \color{brown} return \color{blue}\textit{false} \color{black}\\
\indent (3.4) \color{purple}for all \color{black}$a\in A$, insert $(\tau_a(u), \tau_a(v))$ into $\textit{todo}$\\
\indent (3.5) insert $(u,v)$ into $\textit{R}$\\
(4) \color{brown}return \color{blue} \textit{true}}

\begin{theo}[inspired from theorem 4.3 in \cite{Bonchi17}]\label{correctness} Whenever $\texttt{HKC}^\infty(x,y)$ terminates, it returns \texttt{true} iff $\ll x \rr = \ll y \rr$.
\end{theo}
\begin{proof}
Observe that $R\subseteq b_\beta(c(R)\cup \text{\il{todo}})$ is an invariant for the while loop at step $\texttt{(3)}$. If $\texttt{HKC}^\infty$ returns $\texttt{true}$ then \texttt{todo} is empty and thus $R\subseteq b_\beta (c(R))$ so $R$ is a bisimulation up-to $c$ that contains $(\delta_x,\delta_y)$ and we already know that this yields $\ll x \rr = \ll y \rr$. If $\texttt{HKC}^\infty$ returns \texttt{false}, it encounters a pair $(u,v)$ such that $\beta_\diamond(u) \neq \beta_\diamond(v)$ for a certain $\diamond \in \{ \1, \0 \}$. There exists a word $w$ such that $u = \tau_w(\delta_x)$ and $v = \tau_w(\delta_y)$. Therefore $\varphi_{\beta}(\delta_x) (w) = \varphi_{\beta}(\tau_w(\delta_x)) (\varepsilon) = \beta_\diamond(u) \neq \beta_\diamond(v) = \varphi_{\beta} (\tau_w(\delta_y))(\varepsilon) = \varphi_{\beta}(\delta_y) (w)$ so $\varphi_\beta(\delta_x) \neq \varphi_\beta(\delta_y)$ and we know that this implies $\ll x \rr \neq \ll y \rr$.
\end{proof}
\noindent Despite the fact that during the determinization the state space always becomes infinite, the following results show that if the initial state space $X$ is finite, then $\texttt{HKC}^\infty$ does terminate.
\begin{theo}[see \cite{Boreale2009}]\label{submodule}
Let $\mathcal{R}$ be a ring and $X$ be a finite set. Let $R \subseteq \mathcal{R}^X \times \mathcal{R}^X$ be a relation and let $(v,v') \in \mathcal{R}^X \times \mathcal{R}^X$ be a pair of vectors. We construct a generating set for a submodule of $\mathcal{R}^X$ by defining $U_R = \{ u - u' \mid (u,u') \in R \}$. Then $(v,v') \in c(R)$ iff $v - v' \in [U_R]$.
\end{theo}
\begin{prop}\label{finite} If $X$ is finite, $\texttt{HKC}^\infty(x,y)$ terminates for every $x,y \in \mathbb{R}_\omega^X$.
\end{prop}
\begin{proof}
The set $X$ is finite so $\mathbb{R}_\omega^X = \mathbb{R}^X$. Suppose we have an infinite computation and let $(u_n,v_n)$ be the pair checked at step $n$. The rank of any family of vectors of $\mathbb{R}^X$ is bounded by $|X|<\infty$ and denoted by $\text{\il{rank}}$. When the stationary sequence $(\text{\il{rank}}\{(u_i-v_i)_{i\leq n}\})_{n\in\mathbb{N}}$ has reached its limit at step $N$, then the rest of the sequence is in the submodule of $\mathbb{R}^X$ generated by the first $(u_n-v_n)_{n\leq N}$. According to Theorem \ref{submodule} with $\mathcal{R} = \R$, this means that the $(u_n,v_n)_{n\geq N+1}$ are in the congruence closure of the relation $\{(u_n,v_n)\mid n \leq N\}$. So at each step the test $\texttt{(3.2)}$ succeeds and the cardinality of $\texttt{todo}$ is decreased by $1$. This cardinality is finite as step $N$ so the number of following steps is finite too, yielding a contradiction.
\end{proof}

\begin{theo}[Summary]
For any $\alpha : X \rightarrow P L X$ with $X$ finite, and any $x,y\in X$, $\texttt{HKC}^\infty(x,y)$ terminates and returns $\texttt{true}$ iff $\ll x \rr = \ll y \rr$.
\end{theo}
\subsubsection*{Basic example of $\texttt{HKC}^\infty$}
To begin with, here is a very simple PTS that demonstrates the power of bisimulation up-to congruence.
\begin{center}
\begin{tikzpicture}[->,>=stealth',shorten >=1pt,auto,node distance=2.8cm,
                    semithick]
  \tikzstyle{every state}=[fill=none,draw=black,text=black]

  \node[state](X)   				{$x$};
  \node[state,accepting](F) [right of=X] {$*$};
  \node[state](Y) [right of=F]			{$y$};

  \path (X) edge [loop left] node {a,1/2} (X)
		(X) edge 			 node {1/2} (F)
		(Y) edge [loop right] node {a,1/2} (Y)
		(Y) edge			 node {1/2} (F);
  
\end{tikzpicture}
\end{center}
If we try to compute \texttt{HK($\delta_x, \delta_y$)} on the determinized automaton, the algorithm encounters the pairs $(\delta_x/2^k, \delta_y/2^k),k\geq 0$ and never stops. On the other hand, $\texttt{HKC}^\infty(x,y)$ stops after one step because it has immediately spotted that $(\delta_x/2, \delta_y/2)$ is in the congruence closure of the relation $\{(\delta_x, \delta_y)\}$.
\subsubsection*{Example of $\texttt{HKC}^\infty$ with bisimulation up-to as an accelerant}
Let us compute $\texttt{HKC}^\infty$ to know if in the following automaton the states $x$ and $z$ generate the same probability measure. Here $A = \{ a \}$.
\begin{center}
\begin{tikzpicture}[->,>=stealth',shorten >=1pt,auto,node distance=2.8cm,
                    semithick]
  \tikzstyle{every state}=[fill=none,draw=black,text=black]

  \node[state](X)                    {$x$};
  \node[state](Y) 		[right of=X] {$y$};
  \node[state](Z)		[right of=Y] {$z$};
  \node[state,accepting](F) [below of=Y] {$*$};
  \node[state](I)		[above of=Y]{$i$};

  \path (X) edge              node {a,1/6} (Y)
		(X) edge              node {1/3} (F)
		(X) edge              node {a,1/2} (I)
 		(Y) edge [loop right] node {a,1/3} (Y)
 		(Y) edge              node {2/3} (F)
 		(Z) edge [loop right] node {a,1/3} (Z)
 		(Z) edge 	          node {1/3} (F)
 		(Z) edge              node {a,1/3} (I)
 		(I) edge [loop right] node {a,1} (I);
  
\end{tikzpicture}
\end{center}
Because $X$ is finite, $\R_\omega^X$ has a basis $(e_x,e_y,e_z,e_i)$. An element $u \in \R_\omega^X$ is seen a column vector $ u_x e_x + u_y e_y + u_z e_z + u_i e_i$ in this basis. Moreover $\tilde{\alpha}_\1^\#$ and $\tilde{\alpha}_\0^\#$ are linear forms that can be seen as the row vectors $L_\1 = \begin{pmatrix} 1 & 1 & 1 & 1 \end{pmatrix}$ and $L_\0 = \begin{pmatrix} 1/3 & 2/3 & 1/3 & 0 \end{pmatrix}$, and $\tau_a$ is an endomorphism with a transition matrix $M_a$ defined by $(M_a)_{j,k} = t_a(k)(j)$.
\begin{align*}
u = \begin{pmatrix} u_x\\u_y\\u_z\\u_i\end{pmatrix} & & M_a = \begin{pmatrix} 0 & 0 & 0 & 0 \\ 1/6 & 1/3 & 0 & 0 \\ 0 & 0 & 1/3 & 0 \\ 1/2 & 0 & 1/3 & 1 \end{pmatrix} & & L = \begin{pmatrix} L_\1 \\ L_\0 \end{pmatrix} 
\end{align*}
The algorithm begins with $\texttt{todo} = \{ (\eta_X(x),\eta_X(z)) \} = \{ (e_x,e_z) \}$ and $\texttt{R} = \emptyset$. It checks that $L e_x = L e_z$, etc. as shown in the following table.

\begin{center}
\scalebox{0.8}{
\begin{tabular}{|p{1.1cm}|p{2.8cm}|p{2cm}|p{2.5cm}|p{3.2cm}|p{1.8cm}|}
  \hline
  \centering \texttt{Step} & \centering \texttt{(3.1)} & \centering \texttt{(3.2)} & \centering \texttt{(3.3)} & \centering \texttt{(3.4)} & \centering \texttt{(3.5)}
  \tabularnewline
  \hline

  \centering  \gr{Loop counter} & \centering $(u,v)$ extracted from \texttt{todo} & \centering Check $(u,v) \in c(R)$ & \centering Check $Lu = Lv$ & \centering $(M_a u , M_a v)$ added to $\texttt{todo}$ & \centering Cardinality of $R$
  \tabularnewline
  \hline

  \centering \gr{1} & \centering $( \begin{pmatrix} 1 \\ 0 \\ 0 \\ 0 \end{pmatrix}, \begin{pmatrix} 0 \\ 0 \\ 1 \\ 0 \end{pmatrix} )$ & \centering Fail & \centering $\begin{pmatrix} 1 \\ 1/3 \end{pmatrix} = \begin{pmatrix} 1 \\ 1/3 \end{pmatrix}$ & \centering $( \begin{pmatrix} 0 \\ 1/6 \\ 0 \\ 1/2 \end{pmatrix} , \begin{pmatrix} 0 \\ 0 \\ 1/3 \\ 1/3 \end{pmatrix} )$  & \centering $1$  \tabularnewline \hline

  \centering \gr{2} & \centering $( \begin{pmatrix} 0 \\ 1/6 \\ 0 \\ 1/2 \end{pmatrix} , \begin{pmatrix} 0 \\ 0 \\ 1/3 \\ 1/3 \end{pmatrix} )$ & \centering Fail & \centering $\begin{pmatrix} 2/3 \\ 1/9 \end{pmatrix} = \begin{pmatrix} 2/3 \\ 1/9 \end{pmatrix}$  & \centering $( \begin{pmatrix} 0 \\ 1/18 \\ 0 \\ 1/2 \end{pmatrix} , \begin{pmatrix} 0 \\ 0 \\ 1/9 \\ 4/9 \end{pmatrix})$  & \centering $2$ \tabularnewline
  \hline
  
\centering \gr{3}
& \centering $( \begin{pmatrix} 0 \\ 1/18 \\ 0 \\ 1/2 \end{pmatrix} , \begin{pmatrix} 0 \\ 0 \\ 1/9 \\ 4/9 \end{pmatrix})$
& \centering Success & \centering / & \centering / & \centering 2 \tabularnewline \hline

\centering \gr{4}
& \centering Empty & \centering / & \centering / & \centering / & \centering / \tabularnewline \hline

\end{tabular}
}
\end{center}
The check succeeds in loop $3$ because $(u,v) \in c(R)$ according to theorem \ref{submodule}:
\[ \begin{pmatrix} 0 \\ 1/18 \\ 0 \\ 1/2 \end{pmatrix} - \begin{pmatrix} 0 \\ 0 \\ 1/9 \\ 4/9 \end{pmatrix} = \begin{pmatrix} 0 \\ 1/18 \\ -1/9 \\ 1/18 \end{pmatrix} = \frac{1}{3} \begin{pmatrix} 0 \\ 1/6 \\ -1/3 \\ 1/6 \end{pmatrix} = \frac{1}{3} \left( \begin{pmatrix} 0 \\ 1/6 \\ 0 \\ 1/2 \end{pmatrix} - \begin{pmatrix} 0 \\ 0 \\ 1/3 \\ 1/3 \end{pmatrix} \right) \]
Because $\texttt{todo}$ is eventually empty, the algorithm returns $\texttt{true}$. Indeed, if we compute directly the measures $\ll x \rr$ and $\ll z \rr$, we can see that $\ll x \rr (a^n) = 1/3^{n+1}$, $\ll x \rr (a^\omega) = 1/2$ and same for $\ll z \rr$.
Here the bisimulation up to congruence check is necessary for termination. See what is happening in the determinized infinite automaton.
\begin{center}
\scalebox{0.84}{
\begin{tikzpicture}[->,>=stealth',shorten >=1pt,auto,node distance=2.8cm,
                    semithick]
  \tikzstyle{every state}=[fill=none,draw=none,text=black]

  \node[state](U1)   {$\begin{pmatrix} 1 \\ 0 \\ 0 \\ 0 \end{pmatrix}$};
  \node[state](V1) [below of=U1] {$\begin{pmatrix} 0 \\ 0 \\ 1 \\ 0 \end{pmatrix}$};
  \node[state](U2) [right of=U1]{$\begin{pmatrix} 0 \\ 1/6 \\ 0 \\ 1/2 \end{pmatrix}$};
  \node[state](V2) [right of=V1]{$\begin{pmatrix} 0 \\ 0 \\ 1/3 \\ 1/3 \end{pmatrix}$};
  \node[state](U3) [right of=U2]{$\begin{pmatrix} 0 \\ 1/18 \\ 0 \\ 1/2 \end{pmatrix}$};
  \node[state](V3) [right of=V2]{$\begin{pmatrix} 0 \\ 0 \\ 1/9 \\ 4/9 \end{pmatrix}$};
  \node[state](U4) [right of=U3]{$...$};
  \node[state](V4) [right of=V3]{$...$};
  \node[state](U5) [right of=U4]{$\begin{pmatrix} 0 \\ 1/(2 \times 3^n) \\ 0 \\ 1/2 \end{pmatrix}$};
  \node[state](V5) [right of=V4]{$\begin{pmatrix} 0 \\ 0 \\ 1/3^n \\ (1-3^{-n})/2 \end{pmatrix}$};
  \node[state](U6) [right of=U5]{...};
  \node[state](V6) [right of=V5]{...};

  \path (U1) edge node {a} (U2)
		(V1) edge node {a} (V2)
		(U2) edge node {a} (U3)
		(V2) edge node {a} (V3)
		(U3) edge node {a} (U4)
		(V3) edge node {a} (V4)
		(U4) edge node {a} (U5)
		(V4) edge node {a} (V5)
		(U5) edge node {a} (U6)
		(V5) edge node {a} (V6);
 \draw[-,dashed] (U1) -- (V1) (U2) -- (V2);
 \draw[-,dotted] (U3) -- (V3) (U5) -- (V5);
\end{tikzpicture}}
\end{center}
The construction of the bisimulation up to equivalence (dashed + dotted lines) would take an infinite number of steps. But the construction of the bisimulation up to congruence (dashed lines) takes only $2$ steps.
\subsubsection*{Example of $\texttt{HKC}^\infty$ where infinite traces matter}
Get back to the example of the introduction, taking $A = \{a,b\}$.
\begin{center}
\begin{tikzpicture}[->,>=stealth',shorten >=1pt,auto,node distance=2.8cm,
                    semithick]
  \tikzstyle{every state}=[fill=none,draw=black,text=black]

  \node[state](A)                    {$y$};
  \node(E) [right of=A] {};
  \node[state](B) [right of=E]		{$z$};
  \path (A) edge  [loop left] node {a,1/2} (A)
  		(A) edge  [loop right] node {b,1/2} (A)
  		(B) edge  [loop left] node {a,3/4} (B)
  		(B) edge  [loop right] node {b,1/4} (B);
  		
\end{tikzpicture}
\end{center}
Take $(e_y,e_z)$ as a basis of $\R^X$. Then
\begin{align*}
u = \begin{pmatrix} u_y \\ u_z \end{pmatrix} & & M_a = \begin{pmatrix} 1/2 & 0 \\ 0 & 3/4 \end{pmatrix} & & M_b = \begin{pmatrix} 1/2 & 0 \\ 0 & 1/4 \end{pmatrix} & & L = \begin{pmatrix} 1 & 1 \\ 0 & 0 \end{pmatrix}
\end{align*}
In the first loop, everything is fine because $L \begin{pmatrix} 1 \\ 0 \end{pmatrix} = \begin{pmatrix} 1 \\ 0 \end{pmatrix} = L \begin{pmatrix} 0 \\ 1 \end{pmatrix}$. Then $\texttt{todo} = \left\{ \left(\begin{pmatrix} 1/2 \\ 0 \end{pmatrix}, \begin{pmatrix} 0 \\ 3/4 \end{pmatrix}\right) , \left( \begin{pmatrix} 1/2 \\ 0 \end{pmatrix} , \begin{pmatrix} 0 \\ 1/4 \end{pmatrix} \right) \right\}$ so both possible next equality checks fail: $\begin{pmatrix} 1/2 \\ 0 \end{pmatrix} \neq \begin{pmatrix} 3/4 \\ 0 \end{pmatrix}$ and $\begin{pmatrix} 1/2 \\ 0 \end{pmatrix} \neq \begin{pmatrix} 1/4 \\ 0 \end{pmatrix}$.  Thus, $\texttt{HKC}^\infty$ returns $\texttt{false}$. Remark that not caring about the first line of $L$ amounts to use $\texttt{HKC}$, which returns $\texttt{true}$ because it does not take into account infinite words.

\section{General case}
In this section, we generalize the trace semantics previously defined. We work with continuous PTS, defined later as coalgebras for the analogue of functor $PL$ in the category $\gr{Meas}$. The underlying distributive law is brought to light, so that the origin of the determinization process is better understood. The following table sums up the analogies and differences with the discrete case.
\begin{small}
\begin{center}
\begin{tabular}{c|c|c}
& Discrete case & General case \\
\hline
Category & $\gr{Sets}$ & $\gr{Meas}$ \\
Usual operation & $\sum$ & $\int$ \\
Language functor & $L$ & Measurable version of $L$ \\
Machine functor & $F$ & Measurable version of $F$ \\
Probability monad & Probability monad $P$ & Giry's monad $\P$ \\
Determinization monad & Sub-probability monad $D$ & Sub-Giry's monad $\D$ \\
Initial state space & Set $X$ & Measurable space $(X,\Sigma_X)$ \\
Determinized state & Finitely supported vector & Finite measure ($\leq 1$) \\
Transitions & Matrix $t_a : X \times X \to \I$ & Kernel $t_a : X \times \Sigma_X \to \I$ \\
Final $F$-coalgebra & $\omega$ & Measurable version of $\omega$ \\
Measure coalgebra & $\Pi$ & Measurable version of $\Pi$ \\
Pseudo-final morphism & $\llbracket - \rrbracket : DX \to \M (A^\infty)$ & $\llbracket - \rrbracket : \D X \to \D A^\infty$
\end{tabular}
\end{center}
\end{small}

\noindent In this section we work in the category $\gr{Meas}$ of measurable spaces and functions. It is easy to adapt functors $L$ and $F$, but considering the monads we will need some additional measure-theoretic background.\\\\
Given measurable spaces $X,Y$ and a measurable function $f : X \to Y$, define $LX = A \times X + 1$ along with its $\sigma$-algebra $\Sigma_{LX} = \mathcal{P}(A) \otimes \Sigma_X \oplus \mathcal{P}(1)$, and $Lf = id_A \times f + id_1$. Moreover, define $FX = \I \times \I \times X^A$ along with its $\sigma$-algebra $\mathcal{B}(\I) \otimes \mathcal{B}(\I) \otimes \bigotimes_{a\in A} \Sigma_X$ and $Ff = id_\I \times id_\I \times f^A$.

\subsubsection*{Integration}
Let $(X,\Sigma_X,m)$ be a measure space and $f : X \to \mathbb{R}$ be a measurable function. If $f(X) = \{ \alpha_1,...\alpha_n \}$ for some $\alpha_1,...\alpha_n \in \mathbb{R}_+$, then $f$ is called a \il{simple} function and its integral can be defined as $\int_X f dm = \sum_{i=1}^n \alpha_i m(f^{-1}(\{ \alpha_i \}))$. If $f\geq 0$ pointwise, define $\int_X f dm = \sup \left\{ \int_X g dm \mid g \leq f \text{, $g$ simple}\right\} \in [0,\infty]$. Finally, for any $f : X \to \mathbb{R}$, decompose $f = f^+ - f^-$ where $f^+ \geq 0$ and $f^- \geq 0$. If their integrals are not both $\infty$, define $\int_X f dm = \int_X f^+ dm - \int_X f^- dm$. If this is finite, we say that $f$ is $m$-integrable. Furthermore, for any $S \in \Sigma_X$, the indicator function $\gr{1}_S$ is measurable and we define $\int_S f dm = \int_X \gr{1}_S f dm$.\\\\
Given any measurable function $g : X \to Y$ and any measure $m : \Sigma_X \to \mathbb{R}_+$, the \il{image measure} of $m$ by $g$ is $m \circ g^{-1}$. For any measurable $f : Y \to \mathbb{R}$, $f$ is $m \circ g^{-1}$-integrable iff $f \circ g$ is $m$-integrable and in this case, $\int_Y f d(m\circ g^{-1}) = \int_X (f \circ g) dm$.\\\\
Actually, each positive measurable function $X \to \mathbb{R_+}$ is the pointwise limit of an increasing sequence of simple functions. In order to prove some property for every positive measurable function, one can prove it for simple functions (or for indicator functions, if it is preserved by linear combinations) and show that the property is preserved when taking the limit. Many such proofs use the monotone convergence theorem, which states that if $(f_n)_{n\in\N}$ is an increasing sequence of positive functions with pointwise limit $f$, then $f$ is measurable and $\int_X f dm = \lim \int_X f_n dm$.
\subsubsection*{The Giry monad}
The Giry monad \cite{Giry82} provides a link between probability theory and category theory.
\begin{defi}
In $\gr{Meas}$, the Giry monad $(\P,\eta,\mu)$ is defined as follows. For any measurable space $X$, $\P X$ is the space of probability measures over $(X,\Sigma_X)$, and $\Sigma_{\P X}$ is the $\sigma$-algebra generated by the functions $e^X_S : \P X \to \I$ defined by $e^X_S(m) = m(S)$. For any measurable function $g : X \to Y$, $(\P g)(m) = m \circ g^{-1}$.\\
The unit is defined by $\eta_X(x)(S) = \gr{1}_S(x)$ and the multiplication by $\mu_X(\Phi)(S) = \int_{\P X} e_S^X d\Phi$.
\end{defi}
\noindent In the same way, one can define the sub-Giry monad $(\D,\eta,\mu)$. The only difference is that $\D X$ is then the space of sub-probability measures over $(X,\Sigma_X)$. There is a natural transformation $\iota : \P \Rightarrow \D$ that comes from the inclusion : $\iota_X(m)=m$.

\subsection{Trace semantics via determinization}
The aim of this section is to define trace semantics for continuous PTS, i.e., coalgebras of the form $\alpha : X \to \P L X$. We proceed in the same way as for discrete systems.
\begin{enumerate}
\item[{\color{blue}(i)}] Transform $\alpha$ into a more convenient coalgebra $\tilde{\alpha} : X \to F \D X$.
\item[{\color{red}(ii)}] Determinize $\tilde{\alpha}$ into an $F$-coalgebra $\tilde{\alpha}^\# : \D X \to F\D X$.
\item[(iii)] Factorize the final morphism : $\varphi_{\tilde{\alpha}^\#} = \varphi_\Pi \circ \llbracket - \rrbracket$ and take $\ll - \rr = \llbracket - \rrbracket \circ \eta_X$.
\end{enumerate}
The following diagram sums up the construction. Here $\Omega = (\I \times \I)^{A^*}$ and $\Sigma_\Omega$ is the $\Sigma$-algebra generated by the functions $L \mapsto L(w)$. 
\begin{center}
\begin{tikzpicture}
  \matrix (m) [matrix of math nodes,row sep=2em,column sep=4em,minimum width=2em]
  {
     X & \D X & \D A^\infty & \Omega \\
     \P L X & & & \\
     \D L X & F \D X & F \D A^\infty & F\Omega \\};
  \path[-stealth]
    (m-1-1) edge [color=blue] node [left] {$\alpha$} (m-2-1)
    (m-2-1) edge [color=blue] node [left] {$\iota_{LX}$} (m-3-1)
    (m-1-1) edge [color=red] node [below] {$\eta_X$} (m-1-2)
    (m-1-2) edge [color=black] node [below] {$\llbracket - \rrbracket$} (m-1-3)
    (m-1-3) edge [color=black] node [below] {$\varphi_{\Pi}$} (m-1-4)
    (m-1-2) edge [color=red] node [right] {$\tilde{\alpha}^\#$} (m-3-2)
    (m-1-3) edge [color=black] node [right] {$\Pi$} (m-3-3)
    (m-1-4) edge [color=red] node [right] {$\omega$} (m-3-4)
    (m-3-1) edge [color=blue] node [below] {$\e_X$} (m-3-2)
    (m-3-2) edge [color=black] node [above] {$F\llbracket - \rrbracket$} (m-3-3)
    (m-3-3) edge [color=black] node [above] {$F\varphi_{\Pi}$} (m-3-4)
    (m-1-2) edge [color=red, bend left] node [above] {$\varphi_{\tilde{\alpha}^\#}$} (m-1-4)
    (m-3-2) edge [color=red, bend right] node [below] {$F\varphi_{\tilde{\alpha}^\#}$} (m-3-4)
    (m-1-1) edge [color=blue] node [above right] {$\tilde{\alpha}$} (m-3-2)
    (m-1-1) edge [color=black, bend left] node [above] {$\ll - \rr$} (m-1-3)

    ;
\end{tikzpicture}
\end{center}
\subsubsection*{(i) Translation: from $\alpha$ to $\tilde{\alpha}$}
\begin{prop}
For any measurable space $X$, the function $\e_X : \D L X \to F \D X$ defined by
\[ \e_X(m) = \langle m(LX), m(1), a \mapsto [S \mapsto m(\{a\}\times S)]\rangle \]
is measurable. Moreover, $\e : \D L \Rightarrow F \D$ is a natural transformation.
\end{prop}
\begin{proof} First, see that this is a measurable function. Note that we can write $\e_X = \langle e^{LX}_{LX}, e^{LX}_1, a \mapsto \phi_a \rangle$, where $\phi_a : \D LX \to \D X$ is defined by $\phi_a(m)(S) = m(\{a\} \times S)$. According to Lemma \ref{usb}, it suffices to prove that the $|A|+2$ components are measurable functions. The first two are, by definition of $\Sigma_{\D X}$. Given $a\in A$, use Lemma \ref{generated} and see that for any $S \in \Sigma_X$, $e_S^X \circ \phi_a = e^{LX}_{\{a\}\times S}$ is measurable because $\{a\} \times S \in \Sigma_{LX}$. For naturality, let $f : X \to Y$ be a morphism. Note that if $B_A \in \Sigma_A$, $B_Y \in \Sigma_Y$ and $B_1 \in \Sigma_1$, then $Lf^{-1} (B_A \times B_Y + B_1) = B_A \times f^{-1}(B_Y) + B_1$. Thus, $\e$ is a natural transformation because
\begin{align*} (\e_Y \circ \D L f)(m) &= \langle m(Lf^{-1}(A\times Y + 1)), m(Lf^{-1}(1)), a \mapsto [S\mapsto m(Lf^{-1}(\{a\}\times S))]\rangle \\
&= \langle m(A\times X + 1), m(1), a \mapsto [S\mapsto m(\{a\}\times f^{-1}(S))]\rangle 
\\&=  \langle e_{LX}^{LX}, e_1^{LX}, a \mapsto \D f \circ \phi_a \rangle (m) = (F\D f \circ \e_X)(m)
\end{align*}
\end{proof} 
\noindent Now take $\tilde{\alpha} = \e_X \circ \iota_{LX} \circ \alpha$. We will write its components $\langle \tilde{\alpha}_\1, \tilde{\alpha}_\0, a \mapsto t_a \rangle$. A direct expression of $\tilde{\alpha}$ is given by $\tilde{\alpha}(x) = \langle \underbrace{\alpha(x)(LX)}_1, \alpha(x)(1), a \mapsto [S \mapsto \alpha(x)(\{a\} \times S)]\rangle$.
\subsubsection*{(ii) Determinization}
The following lemma (inspired from \cite{Silva10}) establishes some conditions under which there is a canonical notion of determinization.
\begin{lem}\label{deter}
Let $\gr{C}$ be a category, $F : \gr{C} \to \gr{C}$ be an endofunctor and $(T,\eta,\mu)$ be a monad on $\gr{C}$. Let $f : X \to TFX$ be a $TF$-coalgebra and $h : TFTX \to FTX$ be an Eilenberg-Moore $T$-algebra. Then there exists a unique $T$-algebra morphism $f^\# : (TX,\mu_X) \to (FTX,h)$ such that $f = f^\# \circ \eta_X$.
\end{lem}
\begin{proof} \il{Uniqueness.} Let $f^\#$ be such a morphism, then, using the first diagram of monads $(i)$ and that $f^\#$ is a $T$-algebra morphism $(ii)$, we get
\[ f^\# \underset{(i)}{=} f^\# \circ \mu_X \circ T\eta_X \underset{(ii)}{=} h \circ Tf^\# \circ T\eta_X = h \circ T(f^\# \circ \eta_X) = h \circ Tf \]
\il{Existence.} Take $f^\# = h \circ Tf$, then, using the naturality of $\mu : TT \Rightarrow T$ $(i)$ and that $h$ is an Eilenberg-Moore $T$-algebra $(ii)$, we get that $f^\#$ is a $T$-algebra morphism because
\[ f^\# \circ \mu_X = h \circ Tf \circ \mu_X \underset{(i)}{=} h \circ \mu_{FTX} \circ TTf \underset{(ii)}{=} h \circ Th \circ TTf = h \circ T(h\circ Tf) = h \circ Tf^\# \]
Furthermore, using the naturality of $\eta : Id_\gr{C} \Rightarrow T$ and that $h$ is an Eilenberg-Moore $T$-algebra,  $f^\# \circ \eta_X = h \circ Tf \circ \eta_X = h \circ \eta_{FTX} \circ f = f$.
\end{proof}
\begin{lem}\label{deter2}
With the same notations as for Lemma \ref{deter}, and given a distributive law $\lambda : T F \Rightarrow F T$, then $h = F\mu_X \circ \lambda_{T X} : TFTX \to FTX$ is an Eilenberg-Moore $T$-algebra.
\end{lem}
\begin{proof}
The first diagram of distributive laws $(i)$ and monads $(ii)$ yields
\[ h \circ \eta_{FTX} = F\mu_X \circ \lambda_{TX} \circ \eta_{FTX} \underset{(i)}{=} F\mu_X \circ F\eta_{TX} \circ id_{FX} = F(\mu_X \circ \eta_{TX}) \underset{(ii)}{=} Fid_{TX} = id_{TFX} \]
Furthermore
\begin{align*}
h \circ Th &= F\mu_X \circ \lambda_{TX} \circ TF \mu_X \circ T\lambda_{TX} \\
&= F\mu_X \circ FT\mu_X \circ \lambda_{TTX} \circ T\lambda_{TX} & \text{(naturality of $\lambda$)} \\
&= F(\mu_X \circ T\mu_X) \circ \lambda_{TTX} \circ T\lambda{TX} \\
&= F(\mu_X \circ \mu_{TX}) \circ \lambda_{TTX} \circ T\lambda_{TX} & \text{(second diagram of the monad)} \\
&= F\mu_X \circ F\mu_{TX} \circ \lambda_{TTX} \circ T\lambda_{TX} \\
&= F\mu_X \circ \lambda_{TX} \circ \mu_{FTX} & \text{(second diagram of distributive laws)} \\
&= h \circ \mu_{FTX}
\end{align*}
\end{proof}
\noindent The next step is to define a distributive law $\lambda : \D F \Rightarrow F \D$  in order to apply Lemmas \ref{deter} and \ref{deter2}. In the following we write $id_{FX} = \langle \pi_X^\1, \pi_X^\0, a \mapsto \pi_X^a \rangle$. Note that $\pi^\epsilon : F \Rightarrow \I$ (for $\epsilon \in \{\0,\1\}$) and $\pi^a : F \Rightarrow Id_\gr{C}$ (for $a\in A$) are natural transformations.
\begin{lem} Let $f : X \to \I$ be a measurable function. Then $g_f : \D X \to \I$ defined by $g_f(m) = \int_X f dm$ is measurable. Furthermore, if $g : \D \I \to \I$ is the measurable function defined by $g(m) = \int_{\I} id_{\I} dm$, we have $g_f \circ \eta_X = f$ and $g_f \circ \mu_X = g \circ \D g_f$. \end{lem}
\begin{proof} \il{Measurability}. If $f = \gr{1}_B$ for some $B \in \Sigma_X$, remark that $g_f(m) = \int_X \gr{1}_B dm = m(B) = e_B^X(m)$ so $g_f$ is measurable. By linearity of the integral, $g_f$ is measurable for any simple function. Finally if $f$ is the pointwise limit of an increasing sequence of simple functions $(f_n)_{n\in \N}$, we have that $g_f(m) = \int_X f dm = \int_X \lim f_n dm = \lim \int_X f_n dm = \lim g_{f_n} (m)$ by the monotone convergence theorem.\\
\il{Equations}. If $f = \gr{1}_B$ for some $B \in \Sigma_X$, then $g_f \circ \eta_X = e_B^X \circ \eta_X = \gr{1}_B = f$ and for any $\Phi \in \D \D X$, $(g_f \circ \mu_X)(\Phi) = \mu_X(\Phi)(B) = \int_{\D X} e_B^X d\Phi = \int_{\I} id_\I d(\D e_B^X (\Phi)) = (g \circ \D g_f)(\Phi)$. Note that both equations are preserved by linear combinations, hence the result is true for any simple function. Finally, if $f$ is the pointwise limit of an increasing sequence of simple functions $(f_n)_{n\in \N}$ for which the result is true, we already know that $g_f$ is the pointwise limit of the increasing sequence $g_{f_n}$. In particular $g_{f_n} \circ \eta_X \to g_f \circ \eta_X$ and $g_{f_n} \circ \mu_X \to g_f \circ \mu_X$. For any $\Phi \in \D \D \I$, $(g \circ \D g_{f_n})(\Phi) = \int_\I id_\I d(\D g_{f_n} (\Phi)) = \int_{\D \I} g_{f_n} d\Phi \to \int_{\D \I} g_f d\Phi = (g\circ \D g_f)(\Phi)$. Hence both equations are true for all measurable $f : X \to \I$.
\end{proof}
\noindent Taking $f = id_\I$ and using this lemma we get that $g : \D \I \to \I$ satisfies $g \circ \eta_\I = id_\I$ and $g \circ \mu_\I = g \circ \D g$, hence $g$ is an Eilenberg-Moore $\D$-algebra. For any object $X$ of $\gr{Meas}$, define $\lambda_X : \D F X \to F \D X$ by
\[ \lambda_X = \langle g \circ \D \pi_X^\1, g \circ \D\pi_X^\0, a \mapsto \D\pi_X^a \rangle \]
This is a measurable function because each component is measurable.
\begin{prop}
Let $X,Y$ be objects of $\gr{Meas}$ and $f : X \to Y$ be a measurable function. The following diagrams commute. Consequently, $\lambda : \D F \Rightarrow F \D$ is a distributive law.
\begin{center}
\begin{tikzpicture}
  \matrix (m) [matrix of math nodes,row sep=3em,column sep=2em,minimum width=2em]
  {  \D F X & \D F Y & FX & FX & \D\D FX & \D F \D X & F\D\D X \\
    F \D X & F\D Y & \D FX & F\D X & \D FX & & F\D X \\};
  \path[-stealth]
  	(m-1-1) edge node [above] {$\D F f$} (m-1-2)
  	(m-1-1) edge node [left] {$\lambda_X$} (m-2-1)
  	(m-1-2) edge node [right] {$\lambda_Y$} (m-2-2)
  	(m-2-1) edge node [below] {$F \D f$} (m-2-2)
    (m-1-3) edge node [above] {$id_{FX}$} (m-1-4)
    (m-1-3) edge node [left] {$\eta_{FX}$} (m-2-3)
    (m-2-3) edge node [below] {$\lambda_X$} (m-2-4)
    (m-1-4) edge node [right] {$F\eta_X$} (m-2-4)
    (m-1-5) edge node [above] {$\D\lambda_X$} (m-1-6)
    (m-1-6) edge node [above] {$\lambda_{\D X}$} (m-1-7)
    (m-2-5) edge node [below] {$\lambda_X$} (m-2-7)
    (m-1-5) edge node [left] {$\mu_{FX}$} (m-2-5)
    (m-1-7) edge node [right] {$F\mu_X$} (m-2-7);
\end{tikzpicture}
\end{center}
\end{prop}
\begin{proof}
\begin{align*}
\lambda_Y \circ \D Ff &= \langle g\circ \D \pi^\1_Y \circ \D F f, g \circ \D \pi^\0_Y \circ \D F f, a \mapsto \D \pi^a_Y \circ \D F f \rangle
\\ &= \langle g\circ \D (\pi^\1_Y \circ F f), g \circ \D( \pi^\0_Y \circ F f), a \mapsto \D( \pi^a_Y \circ F f) \rangle
\\ &= \langle g\circ \D(id_\I \circ \pi^\1_X), g \circ \D(id_\I \circ \pi^\0_X), a \mapsto \D(f\circ \pi^a_X) \rangle
\\ &= \langle g\circ \D\pi^\1_X, g \circ \D\pi^\0_X, a \mapsto \D f \circ \D\pi^a_X \rangle & \text{($\pi^\1$, $\pi^\0$ natural)}
\\ &= F\D f \circ \lambda_X
\end{align*}
\begin{align*}
\lambda_X \circ \eta_{FX} &= \langle g \circ \D\pi^\1_X \circ \eta_{FX}, g \circ \D \pi^\0_X \circ \eta_{FX}, a \mapsto \D\pi^a_X \circ \eta_{FX} \rangle
\\ &= \langle g \circ \eta_\I \circ \pi^\1_X, g \circ \eta_I \circ \pi^\0_X, a \mapsto \eta_X \circ \pi^a_X \rangle & \text{ (naturality of $\eta$)}
\\ &= \langle \pi^\1_X, \pi^\0_X, a \mapsto \eta_X \circ \pi^a_X \rangle & \text{($g$ is an EM algebra)}
\\ &= F\eta_X 
\end{align*}
\begin{align*}
& F\mu_X \circ \lambda_{\D X} \circ \D \lambda_X 
\\ &= F\mu_X \circ \langle g \circ \D\pi^\1_{\D X} \circ \D\lambda_X, g \circ \D\pi^\0_{\D X} \circ \D\lambda_X, a \mapsto \D\pi^a_{\D X} \circ \D \lambda_X \rangle & \text{(definition of $\lambda$)}
\\ &= F\mu_X \circ \langle g \circ \D(\pi^\1_{\D X} \circ \lambda_X), g \circ \D(\pi^\0_{\D X} \circ \lambda_X), a \mapsto \D(\pi^a_{\D X} \circ \lambda_X) \rangle
\\ &= F\mu_X \circ \langle g \circ \D(g \circ \D \pi^\1_X), g \circ \D(g \circ \D \pi^\0_X), a \mapsto \D\D\pi^a_X \rangle & \text{ (definition of $\lambda$)}
\\ &= \langle g \circ \D g \circ \D\D \pi^\1_X, g \circ \D g \circ \D\D \pi^\0_X, a \mapsto \mu_X \circ \D\D\pi^a_X \rangle
\\ &= \langle g \circ \mu_\I \circ \D\D \pi^\1_X, g \circ \mu_\I \circ \D\D\pi^\0_X, a \mapsto \mu_X \circ \D\D\pi^a_X \rangle & \text{($g$ EM algebra)} 
\\ &= \langle g \circ \D \pi^\1_X \circ \mu_{FX}, g \circ \D \pi^\0_X \circ \mu_{FX}, a \mapsto \D\pi^a_X \circ \mu_{FX} \rangle & \text{(naturality of $\mu$)}
\\ &= \lambda_X \circ \mu_{FX}
\end{align*}
\end{proof}
\noindent Let us compute the value of our resulting determinization. Given $\tilde{\alpha} : X \to F\D X$, we recall the notation $\tilde{\alpha} = \langle \tilde{\alpha}_\1, \tilde{\alpha}_\0, a \mapsto t_a \rangle$, then take $h = F\mu_X \circ \lambda_{\D X}$ (Lemma \ref{deter2}) and $\tilde{\alpha}^\# = \tilde{\alpha} \circ \D h$ (Lemma \ref{deter}). We get
\begin{align*}
\tilde{\alpha}^\# &= h \circ \D \tilde{\alpha} \\
&= F\mu_X \circ \lambda_{\D X}\circ \D\tilde{\alpha} \\
&= F\mu_X \circ \langle g \circ \D( \pi^\1_{\D X} \circ \tilde{\alpha}), g \circ \D (\pi^\0_{\D X} \circ \tilde{\alpha}), a \mapsto \D (\pi^a_{\D X} \circ \tilde{\alpha}) \rangle \\
&= \langle g \circ \D \tilde{\alpha}_\1, g \circ \D \tilde{\alpha}_\0, a\mapsto \mu_X  \circ \D t_a \rangle
\end{align*}
\noindent Let $m \in \D X$. This more explicit expression shows that the coalgebra that arises from the determinization is natural in the sense that the components of $\tilde{\alpha^\#}$ are basically obtained by integrating the information provided by $\alpha$.
\begin{align*} \tilde{\alpha}^\#(m) &= \left\langle \int_X \tilde{\alpha}_\1 dm, \int_X \tilde{\alpha}_\0 dm, a\mapsto \left[ S \mapsto \int_X t_a(-)(S) dm\right] \right\rangle \\
&= \left\langle \int_X \alpha(-)(LX) dm, \int_X \alpha(-)(1), a\mapsto \left[ S \mapsto \int_X \alpha(-)(\{a\}\times S) dm \right] \right\rangle \end{align*}
\subsubsection*{(iii) Final coalgebra}
This heavy determinization part allows us to work on $F$-coalgebras, which are nice ones because there exists a final object in $\gr{Coalg}(F)$. 
\begin{prop}
Let $\Omega = (\I \times \I)^{A^*}$ and $\Sigma_\Omega$ be the smallest $\sigma$-algebra that makes the functions $e_w : \Omega \to \I \times \I$ defined by $e_w(L) = L(w)$ measurable for every $w \in A^*$. Let $\omega : \Omega \to F\Omega$ be defined by $\omega(L) = \langle L(\varepsilon), a \mapsto L_a \rangle$. Then $(\Omega,\omega)$ is the final $F$-coalgebra.
\end{prop}
\begin{proof}
First, $\omega$ is measurable. Indeed $\pi_1 \circ e_\varepsilon$ and $\pi_2 \circ e_\varepsilon$ are (where $\pi_i : \I \times \I \to \I$) and for $a\in A$, the function $\phi_a : L \mapsto L_a$ is measurable because $e_w \circ \phi_a = e_{aw}$ is measurable for every $w \in A^*$. Let $\beta = \langle \beta_\1,\beta_\0, a\mapsto \tau_a \rangle : X \to FX$ be an $F$-coalgebra. It is easy to see that there is at most one coalgebra morphism $\varphi$ from $\beta$ to $\omega$, because the commutation of the following diagram yields $\varphi(x)(\varepsilon) = \langle \beta_\1(x), \beta_\0(x)\rangle$ and $\varphi(x)(aw) = \varphi(\tau_a(x))(w)$.
\begin{center}
\begin{tikzpicture}
  \matrix (m) [matrix of math nodes,row sep=3em,column sep=4em,minimum width=2em]
  {
     X & \Omega \\
     FX& F\Omega \\};
  \path[-stealth]
  (m-1-1) edge node [above] {$\varphi$} (m-1-2)
  (m-1-1) edge node [left] {$\beta$} (m-2-1)
  (m-1-2) edge node [right] {$\omega$} (m-2-2)
  (m-2-1) edge node [below] {$F\varphi$} (m-2-2);
\end{tikzpicture}
\end{center}
The only thing to check is that this $\varphi$ is measurable. We prove it by induction on words. The function $e_\varepsilon \circ \varphi = \langle \beta_\1, \beta_\0 \rangle$ is measurable, and if $e_w \circ \varphi$ is measurable then $e_{aw} \circ \varphi = (e_w \circ \varphi) \circ \tau_a$ is measurable too by induction hypothesis.
\end{proof}
\noindent Thus for any $F$-coalgebra $\beta$ the final morphism towards $\omega$, denoted $\varphi_\beta$, gives a canonical notion of semantics. What we want it something slightly more specific that takes into account the way $\tilde{\alpha}^\#$ was built to produce a probability measure in $\D A^\infty$. This is obtained via a pseudo-final coalgebra $\Pi : \D A^\infty \to F\D A^\infty$ as follows.
\begin{prop}
Let $\pi : A^\infty \to L A^\infty$ be defined by $\pi(\varepsilon) = *$ and $\pi(aw) = (a,w)$. This is the final $L$-coalgebra.
\end{prop} 
\begin{proof}
The $\sigma$-algebra on $LA^\infty$ is generated by the sets $1$ and $\{a\} \times S$ for $S \in \Sigma_{A^\infty}$. Applying Lemma \ref{usualsets} $(v)$, see that $\pi^{-1}(1) = \{ \varepsilon \} \in \Sigma_{A^\infty}$ and $\pi^{-1}(\{a\} \times S) = aS \in \Sigma_{A^\infty}$. Because of Lemma \ref{generated} this shows that $\pi$ is measurable. Let $\gamma : X \to LX$ be an $L$-coalgebra. There is at most one coalgebra morphism from $\gamma$ to $\pi$. Indeed, the commutation of the following diagram yields that $\pi(\varphi(x)) = (id_A \times \varphi + id_1)(\gamma(x))$ so if $\gamma(x) = *$ then $\varphi(x) = \varepsilon$, and if $\gamma(x) = (a,y)$ then $\varphi(x) = a\cdot \varphi(y)$.
\begin{center}
\begin{tikzpicture}
  \matrix (m) [matrix of math nodes,row sep=3em,column sep=4em,minimum width=2em]
  {
     X & A^\infty \\
     LX & LA^\infty \\};
  \path[-stealth]
  (m-1-1) edge node [above] {$\varphi$} (m-1-2)
  (m-1-1) edge node [left] {$\gamma$} (m-2-1)
  (m-1-2) edge node [right] {$\pi$} (m-2-2)
  (m-2-1) edge node [below] {$F\varphi$} (m-2-2);
\end{tikzpicture}
\end{center}
Using Lemma \ref{generated}, we check that $\varphi$ is measurable by focusing on sets in $S_\infty$. First see that $\varphi^{-1}(\{\varepsilon\}) = \gamma^{-1}(1) \in \Sigma_X$ because $\gamma$ is measurable and that $\varphi^{-1}(\varepsilon A^\infty) = X \in \Sigma_X$. Assume that $\varphi^{-1}(\{w\})$ and $\varphi^{-1}(wA^\infty) \in \Sigma_X$, then $\varphi^{-1}(\{aw\}) = \gamma^{-1}(\{a\}\times \varphi^{-1}(\{w\}))$ and $\varphi^{-1}(awA^\infty) = \gamma^{-1}(\{a\} \times \varphi^{-1}(wA^\infty))$ are in $\Sigma_X$ because $\gamma$ is measurable.
\end{proof}
\noindent Let $\Pi : A^\infty \to FA^\infty$ be the $F$-coalgebra $\Pi = \e_{A^\infty} \circ \D \pi$. It has a direct expression involving the measure derivative; it is exactly the same as the $\Pi$ of section $2$.
\begin{align*}
\Pi(m) &= \langle m(\pi^{-1}(LA^\infty)), m(\pi^{-1}(1)), a \mapsto [S \mapsto m(\pi^{-1}(\{a\} \times S))] \rangle \\
 &= \langle m(A^\infty), m(\varepsilon), a \mapsto m_a \rangle
\end{align*}
The aim is now to factorize the semantics obtained via $\omega$ into semantics obtained via $\Pi$. The following result is kind of a completeness property for this operation.
\begin{lem} \label{injective}
The final morphism $\varphi_\Pi$ from $\Pi$ to $\omega$ is injective.
\end{lem}
\begin{proof}
For any $m, m' \in \D A^\infty$, in order to have $m=m'$, it is sufficient to prove that $m_{|S_\infty} = m'_{|S_\infty}$ according to Theorem \ref{key}. By induction on $w$, we prove that for all $m,m' \in \D A^\infty$ such that $\varphi_\Pi(m) = \varphi_\Pi(m')$, then $\langle m(wA^\infty),m(w) \rangle = \langle m'(wA^\infty), m'(w) \rangle$. First, $\langle m(\varepsilon A^\infty), m(\varepsilon) \rangle = \varphi_\Pi (m)(\varepsilon) = \varphi_\Pi (m')(\varepsilon) = \langle m'(\varepsilon A^\infty), m'(\varepsilon) \rangle$. Note that $\varphi_\Pi(m) = \varphi_\Pi(m')$ implies $\varphi_\Pi(m_a)(w) = \varphi_\Pi(m)(aw) = \varphi_\Pi(m')(aw) = \varphi_\Pi (m'_a)(w)$ so that $\varphi_\Pi(m_a) = \varphi_\Pi(m'_a)$. Use the induction hypothesis to see that $\langle m(awA^\infty), m(aw) \rangle = \langle m_a(wA^\infty), m_a(w) \rangle = \langle m'_a(wA^\infty), m'_a(w) \rangle = \langle m'(awA^\infty), m'(aw) \rangle$. This achieves the induction, so $m$ and $m'$ coincide on $S_\infty$, hence $m=m'$.
\end{proof}
\noindent The following proposition states precisely in which cases the factorization can be done. This is a variant of Theorem \ref{key} in which we really see that our system is making one step. This version is a bit higher-end than Theorem \ref{keybisd} because it also proves that the involved functions are measurable.
\begin{theo} \label{keybis}
Let $\beta = \langle \beta_\1, \beta_\0, a \mapsto \tau_a \rangle : Y \to F Y$ be an $F$-coalgebra. The two following conditions are equivalent:
\begin{enumerate}
\item[$(i)$] There exists an $F$-coalgebra morphism $\llbracket - \rrbracket$ from $\beta$ to $\Pi$.
\item[$(ii)$] The equation $\beta_\1 = \beta_\0 + \sum_{a\in A} \beta_\1 \circ \tau_a$ holds.
\end{enumerate}
In this case, this morphism is unique.
\end{theo}
\noindent For convenience we will now denote $e_S^{A^\infty} \circ \llbracket - \rrbracket$ by $\llbracket - \rrbracket(S)$, and $\phi_a \circ \llbracket - \rrbracket$ by $\llbracket - \rrbracket_a$, where the measure derivative function $\phi_a : m \mapsto m_a$ is measurable as a component of $\Pi$.
\begin{proof}
$(i) \Rightarrow (ii)$ Assume that $\llbracket - \rrbracket$ is a coalgebra morphism from $\beta$ to $\Pi$. Commutation of the diagram yields $\langle \beta_\1, \beta_\0, a \mapsto \llbracket - \rrbracket \circ \tau_a \rangle = \langle\llbracket - \rrbracket(A^\infty), \llbracket - \rrbracket(\varepsilon), a \mapsto \llbracket - \rrbracket_a \rangle $. Let $y \in Y$. Because $\llbracket y \rrbracket$ is a measure, $\beta_\1(y) = \llbracket y \rrbracket (\varepsilon A^\infty) = \llbracket y \rrbracket (\varepsilon) + \sum_{a\in A} \llbracket y \rrbracket (bA^\infty)$. Thus $\beta_\1(y) = \beta_\0 (y) + \sum_{a\in A} \llbracket \tau_a(y) \rrbracket(A^\infty) = \beta_\0(y) + \sum_{a\in A} (\beta_\1 \circ \tau_a) (y)$.\\\\
\il{Uniqueness}. If $\llbracket - \rrbracket'$ is another such morphism, we have $\llbracket - \rrbracket(A^\infty) = \llbracket - \rrbracket'(A^\infty)$, $\llbracket - \rrbracket(\varepsilon) = \llbracket - \rrbracket'(\varepsilon)$  and for any $a \in A$, $\llbracket - \rrbracket \circ \tau_a = \llbracket - \rrbracket_a$ and $\llbracket - \rrbracket' \circ \tau_a = \llbracket - \rrbracket'_a$. An immediate induction yields $\llbracket - \rrbracket_{|S_\infty} = \llbracket - \rrbracket'_{|S_\infty}$, thus $\llbracket - \rrbracket = \llbracket - \rrbracket'$ by Theorem \ref{key}.\\\\
$(ii)\Rightarrow (i)$ Assume that $(ii)$ holds. Let us define $\llbracket - \rrbracket_{|S_\infty}$ by induction:
\begin{align*}
&\llbracket y \rrbracket_{|S_\infty} (\varepsilon A^\infty) = \beta_\1 (y) & & \llbracket y \rrbracket_{|S_\infty} (\varepsilon) = \beta_\0 (y) \\
&\llbracket y \rrbracket_{|S_\infty} (awA^\infty) = \llbracket \tau_a(y) \rrbracket_{|S_\infty} (wA^\infty) & & \llbracket y \rrbracket_{|S_\infty} (aw) = \llbracket \tau_a (y) \rrbracket_{|S_\infty} (w)
\end{align*}
We must prove that it can be extended to a measure, using Theorem \ref{key}. First, note that $\llbracket y \rrbracket_{|S_\infty} (\varepsilon A^\infty) = \beta_\1(y) = \beta_\0(y) + \sum_{a\in A} (\beta_\1 \circ \tau_a)(y) = \llbracket y \rrbracket_{|S_\infty} (\varepsilon) + \sum_{a\in A} \llbracket y \rrbracket_{|S_\infty} (aA^\infty)$. If it is known that for all $y\in Y$, $\llbracket y \rrbracket_{|S_\infty} (wA^\infty) = \llbracket y \rrbracket_{|S_\infty} (w) + \sum_{a\in A} \llbracket y \rrbracket_{|S_\infty} (waA^\infty)$ then for any $b\in A$ we have $\llbracket y \rrbracket_{|S_\infty} (bwA^\infty) = \llbracket \tau_b(y) \rrbracket_{|S_\infty} (wA^\infty) = \llbracket \tau_b(y) \rrbracket_{|S_\infty} (w) + \sum_{a \in A} \llbracket \tau_b(y) \rrbracket_{|S_\infty} (waA^\infty) = \llbracket y \rrbracket_{|S_\infty} (bw) + \sum_{a\in A} \llbracket y \rrbracket_{|S_\infty} (bwaA^\infty)$. This proves the $(ii)$ of Theorem \ref{key}. We denote by $\llbracket - \rrbracket$ the extension of $\llbracket - \rrbracket_{|S_\infty}$. We postpone the proof of the measurability of $\llbracket - \rrbracket$; what is left is the commutation of the coalgebra diagram. The first line of the definition of $\llbracket - \rrbracket_{|S_\infty}$ gives directly that $\beta_\1 = \llbracket - \rrbracket(A^\infty)$ and $\beta_\0 = \llbracket - \rrbracket(\varepsilon)$. Let $a\in A$. For any $y \in Y$, according to the second line of the definition of $\llbracket - \rrbracket_{|S_\infty}$, the measures $\llbracket \tau_a(y) \rrbracket$ and $\llbracket y \rrbracket_a$ coincide on $S_\infty$, hence are equal according to Theorem \ref{key}, so $\llbracket - \rrbracket \circ \tau_a = \llbracket - \rrbracket_a$. This achieves the proof that the diagram commutes.\\\\
\il{Measurability}. It is not immediate to notice why $\llbracket - \rrbracket : Y \to \D A^\infty$ is a measurable function. What has to be shown according to Lemma \ref{usb} is that for any $S \in \Sigma_{A^\infty}$, $\llbracket - \rrbracket(S)$ is measurable. This is true when $S \in S_\infty$. Indeed, $\llbracket - \rrbracket(\emptyset)$ is the zero function, which is measurable. For the rest we proceed by induction. Obviously $\llbracket - \rrbracket(\varepsilon A^\infty) = \beta_\1$ and $\llbracket - \rrbracket(\varepsilon) = \beta_\0$ are measurable because $\beta$ is. Furthermore, $\llbracket - \rrbracket(awA^\infty) = \llbracket - \rrbracket_a(wA^\infty) = \llbracket - \rrbracket(wA^\infty) \circ \tau_a$ and $\llbracket - \rrbracket(aw) = \llbracket - \rrbracket_a(w) = \llbracket - \rrbracket(w) \circ \tau_a$ are measurable by induction hypothesis and composition.\\\\
Let $Z$ be a set. A set $P \subseteq \mathcal{P}(Z)$ is a \il{$\pi$-system} if it is non-empty and closed under finite intersections. A set $D \subseteq \mathcal{P}(Z)$ is a \il{$\lambda$-system} if it contains $Z$ and is closed under difference (if $A,B \in D$ and $A\subseteq B$ then $B\setminus A \in D$) and countable \il{increasing} union. A widely known theorem of measure theory, namely the $\pi-\lambda$ theorem (see \cite{Aliprantis06}, lemma $4.11$) is that given $P$ a $\pi$-system, $D$  a $\lambda$-system such that $P \subseteq D$, then $\sigma_Z(P) \subseteq D$.\\\\
Take $Z = A^\infty$, $P = S_\infty$ and $D = \{ S \in \Sigma_{A^\infty} \mid \llbracket - \rrbracket(S) \text{ is measurable} \}$. It is easy to see that $S_\infty$ is a $\pi$-system. Moreover, $D$ is a $\lambda$-system. Indeed, $A^\infty \in D$ (see above), if $(S_n)_{n\in\N}$ is an increasing sequence of sets in $D$, then $\llbracket - \rrbracket(S_1 \setminus S_0) = \llbracket - \rrbracket(S_1) - \llbracket - \rrbracket(S_0)$ is measurable as a difference of measurable functions and $\llbracket - \rrbracket \left( \bigcup_{n\in\N} S_n\right) = \lim_{n\to \infty} \llbracket - \rrbracket (S_n)$ is measurable as a pointwise limit of measurable functions. Finally, given the preceding paragraph, we have $S_\infty \subseteq D$. The $\pi-\lambda$ theorem therefore yields $\Sigma_{A^\infty} \subseteq D$.
\end{proof}
\noindent An interpretation of the last proposition is that, in the subcategory of $F$-coalgebras that satisfy the equation $(ii)$, the final object is $\Pi$. If Proposition \ref{keybis} holds, then note that $\varphi_{\Pi} \circ \llbracket - \rrbracket$ is a coalgebra morphism from $\beta$ into the final coalgebra $\omega$. Hence by finality $\varphi_{\Pi} \circ \llbracket - \rrbracket = \varphi_{\beta}$. This is kind of a soundness property for our factorization. Soundness and completeness together yield the following proposition, which is exactly the same as in section $2$.
\begin{prop}
Let $\beta : Y \to FY$ be an $F$-coalgebra for which Proposition \ref{keybis} holds. Then for any $y,z \in Y$, $\llbracket y \rrbracket = \llbracket z \rrbracket$ iff $\varphi_\beta(y) = \varphi_\beta(z)$.
\end{prop} 
\begin{proof}
By Lemma \ref{injective}, $\llbracket y \rrbracket = \llbracket z \rrbracket$ iff $(\varphi_\Pi \circ \llbracket - \rrbracket)(y) = (\varphi_\Pi \circ \llbracket - \rrbracket)(z)$ iff $\varphi_\beta(y) = \varphi_\beta(z)$.
\end{proof}
\noindent Back to $\alpha : X \to \P L X$ we check that Proposition \ref{keybis} holds for $\tilde{\alpha}^\# = \langle \tilde{\alpha}^\#_\1,\tilde{\alpha}^\#_\0, a \mapsto \tau_a \rangle$. Note that because $\alpha(x)(LX) = 1$, we have for all $m \in \D X$ that $m(X) = \int_X 1dm = \int_X \alpha(x)(LX) dm = \tilde{\alpha}^\#_\1(m)$. This justifies the last equality:
\begin{align*}
\tilde{\alpha}^\#_\1(m) &= \int_X \alpha(-)(LX) dm = \int_X \left(\alpha(-)(1) + \sum_{a\in A} \alpha(-)(\{a\} \times X) \right) dm
 \\ &= \int_X \alpha(-)(1) dm + \sum_{a\in A} \int_X \alpha(-)(\{a\} \times X) dm \\
 &= \tilde{\alpha}^\#_\0(m) + \sum_{a\in A} \tau_a(m)(X) = \tilde{\alpha}^\#_\0(m) + \sum_{a\in A} (\tilde{\alpha}_\1^\# \circ \tau_a)(m)
\end{align*}
\gr{Conclusion.} To any $\alpha : X \to \P L X$ can be given a canonical trace semantics via a determinization process. This is a function $\llbracket - \rrbracket : \D X \to \D A^\infty$.
\subsection{Related results}
\subsubsection*{Link with Kerstan's trace semantics}
In \cite{Kerstan13}, given an $\alpha : X \to \P L X$, the trace semantics $\gr{tr} : X \to \P A^\infty$ is defined by
\begin{align*}
&\gr{tr}(x)(\varepsilon A^\infty) = \alpha(x)(LX) \text{\indent} (= 1) & & \gr{tr}(\varepsilon) = \alpha(x)(1) \\
&\gr{tr}(x)(awA^\infty) = \int_X \gr{tr}(-)(wA^\infty) dt_a(x) & & \gr{tr}(x)(aw) = \int_X \gr{tr}(-)(w) dt_a(x)
\end{align*}
We will hereby prove that this sematics fits with ours, in the sense that the following diagram commutes.
\begin{center}
\begin{tikzpicture}
  \matrix (m) [matrix of math nodes,row sep=3em,column sep=4em,minimum width=2em]
  {
     X & \P A^\infty \\
     \D X & \D A^\infty \\};
  \path[-stealth]
  (m-1-1) edge node [above] {$\gr{tr}$} (m-1-2)
  (m-1-1) edge node [left] {$\eta_X$} (m-2-1)
  (m-1-2) edge node [right] {$\iota_{A^\infty}$} (m-2-2)
  (m-2-1) edge node [below] {$\llbracket - \rrbracket$} (m-2-2);
\end{tikzpicture}
\end{center}
\noindent Define $\ll - \rr = \llbracket - \rrbracket \circ \eta_X$ as in the discrete case.
\begin{lem}
For any $m \in \D X$ and any $S\in S_\infty$, $\llbracket m \rrbracket(S) = \int_X \ll - \rr(S) dm$.
\end{lem}
\begin{proof}
In this proof there may be times when $\int_X f dm$ is denoted by $\int_{x\in X} f(x) m(dx)$. First, show that for any measurable function $f : X \to \I$,
\[\int_X f d\tau_a (m) = \int_{x \in X} \left( \int_X f d t_a (x) \right) m(dx) \]
Using a density argument, first look at the case $f = \gr{1}_B$ for some $B \in \Sigma_X$. The equality becomes $\tau_a(m)(B) = \int_X t_a(-)(B) dm$, which is true by definition of $\tau_a$. Furthermore, the property is clearly preserved by linear combination, so it is true for simple functions. Now let $(f_n)_{n\in\N}$ be an increasing sequence of simple functions with pointwise limit $f$. Then $\int_X f dt_a(m) = \lim \int_X f_n dt_a(m) = \lim \int_{x \in X} \left(\int_X f_n dt_a(x)\right) m(dx) = \int_{x\in X} \left(\int_X \lim f_n dt_a(x)\right) m(dx) = \int_{x\in X} \left(\int_X f dt_a(x)\right) m(dx)$. The exchanges between limit and integral are justified by the monotone convergence theorem.\\\\
Note further that $\ll x \rr (\varepsilon A^\infty) = (\tilde{\alpha}^\#_\1 \circ \eta_X)(x) = \tilde{\alpha}_\1(x) = \alpha(x)(LX)$ and in the same way $\ll x \rr (\varepsilon) = \alpha(x)(1)$. Now let us prove the lemma by induction, for all $m\in \D X$. First \begin{align*}
& \llbracket m \rrbracket(\varepsilon A^\infty) = \tilde{\alpha}^\#_\1(m) = \int_X \alpha(-)(LX) dm = \int_X \ll - \rr (\varepsilon A^\infty) dm \\
& \llbracket m \rrbracket(\varepsilon) = \tilde{\alpha}^\#_\0(m) = \int_X \alpha(-)(1) dm = \int_X \ll - \rr (\varepsilon) dm 
\end{align*}
Assume the result is true for $wA^\infty$ and $w$. Take $\diamond \in \{\{\varepsilon\},A^\infty \}$.
\begin{align*} 
\llbracket m \rrbracket (aw\diamond) &= \llbracket \tau_a(m) \rrbracket (w\diamond) = \int_X \ll - \rr (w\diamond) d\tau_a(m) & \text{(induction hypothesis)}
\\ &= \int_{x\in X} \left( \int_X \ll - \rr (w\diamond) dt_a(x) \right) m(dx) & \text{(preliminary lemma)}
\\ &= \int_{x \in X} \llbracket \tau_a(\eta_X(x)) \rrbracket(w\diamond) m(dx) & \text{(definition of $\tau_a$)}
\\ &= \int_X \ll - \rr (aw\diamond) dm
\end{align*}
\end{proof}
\noindent Using this last lemma and that $\tau_a \circ \eta_X = t_a$, we have for any $x\in X$: 
\begin{align*}
&\ll x \rr (\varepsilon A^\infty) = \alpha(x)(LX) \\
&\ll x \rr (\varepsilon) = \alpha(x)(1) \\
&\ll x \rr (awA^\infty) = \llbracket (\tau_a \circ \eta_X)(x) \rrbracket (wA^\infty) = \llbracket t_a(x) \rrbracket (wA^\infty) = \int_X \ll - \rr(wA^\infty) dt_a(x) \\
&\ll x \rr (aw) = \llbracket (\tau_a \circ \eta_X)(x) \rrbracket (w) = \llbracket t_a(x) \rrbracket (w) = \int_X \ll - \rr(w) dt_a(x)
\end{align*}
Thus, for any $x\in X$, $\ll - \rr(x)$ and $(\iota_{A^\infty} \circ \gr{tr})(x)$ are measures in $\D A^\infty$ that coincide on $S_\infty$. Because of Theorem \ref{key}, they are equal. Consequently the above diagram commutes, which mean that the trace semantics we get via determinization and Eilenberg-Moore algebras is the same as the Kleisli trace semantics of \cite{Kerstan13}.
\begin{prop} The two trace semantics denoted by $\ll - \rr$ and $\gr{tr}$ coincide. \end{prop}
\subsubsection*{Link with the discrete case}
In the event that $\alpha : X \to \P L X$ can be seen as a discrete system, i.e., for all $x \in X$, $\alpha(x)$ is a linear sum of Dirac distributions, then the general semantics coincide with those obtained in section $2$.

\subsubsection*{Link with a more general correspondence}
In \cite{Jacobs15}, an abstract link is established between Kleisli trace semantics and determinized trace semantics. It turns out that when both constructions are possible, and under some compatibility conditions, the two trace semantics can be compared. In the setting of our paper, it is proved above that indeed the trace semantics are the same. But we do not know if the general result can be directly applied here for at least one reason: the correspondence stated in \cite{Jacobs15} uses only one monad for both constructions. This allows to relate both constructions in an easier fashion, via for example the extension natural transformation $\e : \D L \Rightarrow F \D$. The moment we violate this is when we use the injection natural transformation $\iota : \P \Rightarrow \D$. It seems actually impossible to choose to use only $\P$ or $\D$. The Giry monad is necessary because we do need the sums-to-$1$ condition to ensure that Theorem \ref{keybis} is satisfied. The sub-Giry monad is necessary because the components of $\tilde{\alpha}$ do not sum to $1$. One may argue that $\P(A\times X +1)\simeq \D (A \times X)$ (via the function $m \mapsto m_{|{\Sigma_A \otimes \Sigma_X}}$), but this does not solve this issue. Indeed, Kleisli semantics of PTS of the form $X \to \D(A\times X)$ is trivial (see Theorem 3.33 in \cite{Kerstan13}).\\\\
Our work is largely done by hand because there are no general enough results about systems of the shape $X \to \P (A \times X + 1)$. Either this is because such systems are a really specific case and the existence of measure semantics is a little wonder, or this may be because we were not able to see a more stylish way to proceed.

\section*{Conclusion}
\addcontentsline{toc}{section}{Conclusion}
The recent formalization of automata through coalgebras allows to take a step back and understand why a trace semantics is the good one in a certain sense. In addition, it provides some tools such as bisimulation (up-to). We took as a starting point the trace semantics for continuous PTS given by Kerstan in \cite{Kerstan13} and redefined it using a determinization process in both discrete and continuous cases. It seems that a Kleisli approach is inadequate for taking into account infinite traces in $\gr{Sets}$, whereas our determinization approach can do this using a small and very localized amount of measure theory. For discrete PTS, bisimulations up-to turned out to be a fertile ground for finding an algorithm that checks trace equivalence.\\\\
In section $3$ the algorithmic considerations of section $2$ could be adapted, but this would be irrelevant. Indeed, such algorithms for continuous systems will not be computable as they involve infinite sums or integrals, in contrast to $\texttt{HKC}^\infty$ which only needs some matrix multiplications. Moreover, for any interesting general system, i.e., any system that can not be reducted to the discrete case, such algorithms will not terminate. The only case for which they would be useful is to prove that a given general system is equivalent to a given discrete system.\\\\
One could have worked with only one output in the machine functor, as this is usually the case. For example, if the total mass output is dropped, its information is not lost forever as we can compute it by summing the mass on every word using the termination output and the transitions. However, our way of doing is shaped to highlight the step-by-step motion of the automaton and the fact that each state carries two equally important pieces of information. This makes the coalgebraic treatment via the machine functor sound in regards to bisimulation.  Another reason is that using two outputs lead "easily" to the statement about the existence of a pseudo-final morphism into $\D A^\infty$.\vfill
\noindent \gr{Acknowledgements}. I would like to thank Marc Aiguier and Jan Rutten for helping me find this internship. Thanks to Damien Pous, Filippo Bonchi and Jan Rutten again for listening to me and asking inspiring questions. Of course thanks to my supervisor Jurriaan Rot for the same reasons plus his valuable comments, availability and cheerfullness. And lastly, thanks to Meven Bertrand for keeping me good company.
\newpage
\bibliographystyle{plain}
\bibliography{bibli}

\end{document}